\documentclass[10pt, twocolumn, final, journal]{IEEEtran}

\usepackage{setspace}
\usepackage[cmex10]{amsmath}
\usepackage{amsfonts}
\usepackage{subfigure,url}

\usepackage{xcolor}
\usepackage{xkeyval}
\usepackage{tikz}
\usepackage{algorithm} 
\usepackage{algorithmic}
\usepackage[noadjust]{cite}
\usepackage[center]{caption}

\usepackage{bbm}

\newtheorem{theorem}{Theorem}
\newtheorem{lemma}[theorem]{Lemma}
\newtheorem{corollary}[theorem]{Corollary}
\newtheorem{proposition}[theorem]{Proposition}

\newtheorem{example}{Example}
\newtheorem{definition}{Definition}
\newtheorem{assumption}{Assumption}

\newcommand{\argmin}{\operatornamewithlimits{argmin}}

\allowdisplaybreaks[4] 

\ifodd 1
\newcommand{\com}[1]{\textbf{\color{red} (COMMENT: #1)}} 
\newcommand{\response}[1]{\textbf{\color{magenta} (RESPONSE: #1)}} 
\else

\newcommand{\com}[1]{}
\newcommand{\response}[1]{}   
\fi

\begin{document}
\title{Competition of Wireless Providers for Atomic Users}
\author{\authorblockN{Vojislav Gaji\'{c}\IEEEauthorrefmark{1}, Jianwei Huang\IEEEauthorrefmark{2}, and Bixio Rimoldi\IEEEauthorrefmark{1}}\\\authorblockA{\begin{tabular}{cc} \IEEEauthorrefmark{1}Mobile Communications Lab  &             \IEEEauthorrefmark{2}Network Communications and Economics Lab  \\ School of Computer and Communication Sciences & Department  of Information Engineering\\ Ecole Polytechnique F\'ed\'erale de Lausanne & The Chinese University of Hong Kong \\ Lausanne, Switzerland & Shatin, Hong Kong  \\  \{vojislav.gajic, bixio.rimoldi\}@epfl.ch &  jwhuang@ie.cuhk.edu.hk\end{tabular}}}

\maketitle

\begin{abstract}

We study a problem where wireless service providers compete for heterogenous wireless users. The users differ in their utility functions as well as in the perceived quality of service of individual providers. We model the interaction of an arbitrary number of providers and users as a two-stage multi-leader-follower game. We prove existence and uniqueness of the subgame perfect Nash equilibrium for a generic channel model and a wide class of users' utility functions. We show that the competition of resource providers leads to a globally optimal outcome under mild technical conditions. Most users will purchase the resource from only one provider at the unique subgame perfect equilibrium. The number of users who connect to multiple providers at the equilibrium is always smaller than the number of providers. We also present a decentralized algorithm that globally converges to the unique system equilibrium with only local information under mild conditions on the update rates.
\end{abstract}

\section{Introduction}

Due to the deregulation of the telecommunication industry, future wireless users are likely to freely choose a provider (or providers) offering the best tradeoff of parameters. {This is already happening with some public Wi-Fi connections, where users can connect to wireless access points of their choice, with usage-based payments and no contracts.  Despite the common presence of a free public wi-fi network, some users may still choose more expensive providers who offer better quality of service. Another example of the deregulation trend is the analog television (UHF) spectrum which was recently open for unregulated use in the US \cite{FCC:2008uq}. 

In this work, we consider a situation where wireless service providers compete to sell a limited amount of wireless resources (e.g., frequency bands, time slots, transmission power) to users who are free to choose their provider(s). We investigate how providers set prices for the resource, and how users choose the amount of resource they purchase and from which providers. The focus of our study is to characterize the outcome of this interaction. We consider the general case where different users have different utility functions and experience different channel conditions to different service providers. 

We model the user-provider interaction as a multi-leader-follower game (see \cite{Leyffer:2005pi,Wang:2008jt}). The providers announce the wireless resource prices in the first stage, and the users announce their demand for the resource in the second stage. A user may purchase resource from a provider with a poor channel if the price of the resource is low, or conversely from an expensive one if the channel is strong. The providers select their prices to maximize their revenues, keeping in mind the impact of their prices on the demand of the users. 

The contributions of our work are as follows:

\begin{itemize}
\item \emph{General Heterogeneous Wireless Network Model}: We study a general {network} model that captures the heterogeneity of wireless users and service providers. The users have different utility functions, the providers have different resource constraints, the channel gains between users and providers are independent and arbitrarily distributed, and the numbers of users and providers can be arbitrary. 

\item \emph{Unique Socially Optimal Allocation}: We first study the problem of maximizing social welfare under a fairly general utility function model. {We show that when channel parameters are randomly drawn from continuous distributions, there exists a \emph{unique} optimal solution to the problem with probability 1, despite the non-strict convexity of the optimization problem.} 

\item \emph{Existence, Uniqueness, and Zero Efficiency Loss of Equilibrium}: We further prove existence and uniqueness of the subgame perfect Nash equilibrium in the two-stage game, under an easily verifiable sufficient condition on the users' utility functions. Moreover, we show that the unique equilibrium maximizes the social welfare, despite the selfish nature of the providers and users. 

\item \emph{Primal-Dual Algorithm converging to Equilibrium}: We provide a decentralized algorithm that results in an equilibrium of the provider competition game. The participants only need local information during the execution of this algorithm. Providers only need to know the demand of the users, while users only need to consider the prices given by the providers.
\end{itemize}

We begin by describing the provider competition model in Section \ref{sec:model}. In Section \ref{sec:socOpt} we discuss the socially optimal resource allocation and in Section \ref{sec:oligopoly} we analyze the provider competition game. We define the primal-dual update algorithm and prove its convergence in Section \ref{sec:primal-dual}. We provide numerical results and dicussion in Section \ref{sec:discussion}. We discuss the related work in Section \ref{sec:related} and conclude in Section \ref{sec:conclusion}. 

\section{Problem Formulation}

\label{sec:model}
We consider a set $\mathcal{J}=\{1,\ldots,J\}$ of service providers and a set $\mathcal{I}=\{1,\ldots,I\}$ of users. Provider $j\in\mathcal{J}$ maximizes its revenue by selling up to $Q_j$ amount of resource to the users. A user $i \in \mathcal{I}$ maximizes its payoff by purchasing resources from one or more providers. The communication can be both downlink or uplink, as long as users do not interfere with each other by using orthogonal resources. We model the interaction as a multi-leader-follower game (see \cite{Leyffer:2005pi,Wang:2008jt}), where providers are the leaders and users are the followers. We called this game the provider competition game. We assume that the users play the game during the coherence time of their channels. {This can be reasonable for a quasi-static network environment (e.g., users are laptops or smart phones in offices or airports).} This means that the channel gains remain roughly constant and can be made known to all parties. For example, each provider collects its channel condition information to each user, and then broadcasts this information to all users. This assumption will be relaxed in Section \ref{sec:primal-dual}, where we consider a decentralized algorithm that results in the same outcome as the multi-leader-follower game.

{\subsection{Provider competition game}}
The provider competition game consists of two stages. In the first stage, providers announce prices $\boldsymbol{p}=[p_{1}\;\cdots \;p_{J}]$, where $p_{j}$ is the unit resource price charged by provider $j$.  In the second stage, each user $i \in \mathcal{I}$ chooses a demand vector $\boldsymbol{q_{i}}=[q_{i1}\;\cdots\;q_{iJ}]$, where $q_{ij}$ is the demand to provider $j$. We denote by $\boldsymbol{q}=[\boldsymbol{q}_{1} \cdots \boldsymbol{q}_{I}]$ the demand vector of all users. 

In the second stage where prices $\boldsymbol{p}$ are given, the goal of user $i$ is to choose $\boldsymbol{q}_i$ to maximize its payoff, which is utility minus payment: 
\begin{align}
v_{i}(\boldsymbol{q_{i}},\boldsymbol{p})= u_{i}\left( \sum_{j=1}^{J}{q_{ij}}{c_{ij}}\right)-\sum_{j=1}^{J}p_{j}q_{ij},
\label{eqn:utility}
\end{align}
where $c_{ij}$ is the \emph{channel quality offset} for the channel between user $i$ and the base station of provider $j$ (see Example \ref{exa:TDMA} and Assumption \ref{as:channelGainsOffsetFunction}), and $u_i$ is an increasing and concave utility function. In the first stage, a provider $j$ chooses price $p_{j}$ to maximize its revenue $p_j \sum_{i=1}^{I} q_{ij}$ subject to the resource constraint $\sum_{i=1}^{I} q_{ij} \leq Q_j$, by taking into account the demand of the users in the second stage. We consider linear pricing with no price discrimination across the users. 
 
Under this model, a user is allowed to purchase from several providers at the same time. For this to be feasible, a user's device might need to have several wireless interfaces. Mathematically, the solution of this model gives an upper bound on best performance of any situation where users are constrained to purchase from one provider alone. Interestingly, our results show that for most users, i.e. no less than $I-J$, the optimal strategy is to choose exactly one provider.

Next we give a concrete example of how our model is mapped into a physical wireless system. 

\begin{example}
\label{exa:TDMA}
Consider wireless providers operating on orthogonal frequency bands $W_{j},\; j\in \mathcal{J}$. Let $q_{ij}$ be be the fraction of time that user $i$ is allowed to transmit exclusively on the frequency band of provider $j$, with the constraint $\sum_{i\in\mathcal{I}_{j}}q_{ij}=1$, $j \in \mathcal{J}$. Furthermore, assume that each user has a peak power constraint $P_{i}$. We can then define $c_{ij}=W_{j} \log (1+\frac{P_{i} |h_{ij}|^{2}}{\sigma_{ij}^{2}W_{j}})$, where $h_{ij}$ is the channel gain and $\sigma^{2}_{ij}$ is the Gaussian noise variance for the channel between user $i$ and network $j$. In this case, a user's payoff is the difference between its utility function (in terms of total rate) and payments, $v_{i}=u_{i} (\sum_{j=1}^{J}q_{ij} c_{ij})-\sum_{j=1}^{J}p_{j}q_{ij}$. 
\end{example}

Although the $c_{ij}$ channel quality offset factor represents channel capacity in Example \ref{exa:TDMA}, it can be any increasing function of the channel strength  depending on the specific application scenario. 

Finally, we remark that the problem shares certain similarity with the multi-path routing problem in a generalized network flow setting, where each source corresponds to a user and each link corresponds to a provider. The key difference is that in our model the user-provider connections have different weights, which is not the case for the multipath routing problem. 

{\subsection{Model assumptions}}
We make the following assumptions throughout this paper:
\begin{assumption}
\label{as:concaveUtility}
 \emph{(Utility functions):} For every user $i\in \mathcal{I}$, $u_{i}(x)$ is differentiable, increasing, and strictly concave in $x$. This is a standard way to model elastic data applications in network literature (see, e.g., \cite{Kelly:1998}).
\end{assumption}
\begin{assumption}
\label{as:channelGainsOffsetFunction}
{\emph{(Channel quality offsets and channel gains):} Channel quality offsets $c_{ij}$ are drawn independently from continuous, possibly different utility distributions. In particular $Pr(c_{ij}=c_{kl})=0$ for any $i,k\in \mathcal{I},\;j,l\in \mathcal{J}$. The channel quality offset accounts for the effect that buying the same amount of resource from different providers typically has different effects on a user's quality of service. As Example \ref{exa:TDMA} shows, channel quality offset $c_{ij}$ may be a function of the channel gain $h_{ij}$ between user $i$ and provider $j$. In this case the assumption is fulfilled if channel gains are drawn from independent continuous probability distribution (e.g., Ralyleigh, Rician, distance-based path-loss model). 
}

\end{assumption}
\begin{assumption}
\emph{(Atomic and price-taking users):} The demand for an atomic user is not infinitely small and can have an impact on providers' prices. Precise characterization of this impact is one of the focuses of this paper. On the other hand, users are price-takers by the assumption of the two-stage game, and {do} not strategically influence prices. 
\end{assumption}

To analyze the properties of the provider competition game, in Section \ref{sec:socOpt} we study a related socially optimal resource allocation problem and show the uniqueness of its solution in terms of users' demands. Then, in Section \ref{sec:oligopoly}, {we come back to the provider competition game.} We show that {the unique socially optimal solution} solution corresponds to the unique equilibrium of the provider competition game, in which case the selfish and strategic behavior of providers and users leads to zero efficiency loss. 

\section{Social Optimum}
\label{sec:socOpt} 

\subsection{Social welfare maximization}
In this section we consider a social welfare problem, which aims at maximizing the sum of payoffs of all participants,  (users and providers). The social welfare problem is equivalent to maximizing the sum of users' utility functions since the payments between users and providers cancel out. We show the uniqueness of its solution in terms of users' demands. For clarity of exposition, we define the following notation. 
\begin{definition}{\emph{(Effective resource)}}
Let $\boldsymbol{x}=[x_{1} \; \cdots \; x_{I}]$ be the vector of \emph{effective resources}, where $x_{i}(\boldsymbol{q}_i)=\sum_{j=1}^{J}q_{ij}c_{ij}$ is a function of user $i$'s demand $\boldsymbol{q}_{i}=[q_{i1}\ldots q_{iJ}]$.
\end{definition}

The social welfare optimization problem (SWO) is:
\begin{align}
\textbf{SWO}: \max\; u(\boldsymbol{x})=&\sum_{i=1}^{I}u_{i}\left(x_{i}\right) \label{eqn:max} \\
\text{subject to }& \sum_{j=1}^{J}q_{ij} c_{ij}=x_{i} \; i \in \mathcal{I} \label{eqn:relative}\\
& \sum_{i=1}^{I}q_{ij}= Q_{j},\; j \in \mathcal{J} \label{eqn:clearing}\\
\text{over } & q_{ij},x_{i}\geq0 \; \forall i \in \mathcal{I},j \in \mathcal{J}. \label{eqn:positive}
\end{align}

We expressed the SWO in terms of two different variables: effective resource vector $\boldsymbol{x}$ and demand vector $\boldsymbol{q}$, even though the problem can be expressed entirely in terms of $\boldsymbol{q}$. In particular, a vector $\boldsymbol{q}$ uniquely determines a vector $\boldsymbol{x}$ through equations (\ref{eqn:relative}), i.e. we can write $\boldsymbol{x}$ as $\boldsymbol{x}(\boldsymbol{q})$. With some abuse of notation we will write $u(\boldsymbol{q})$ when we mean $u(\boldsymbol{x}(\boldsymbol{q}))$. 

\begin{lemma}
The social welfare optimization problem SWO has a unique optimal solution $\boldsymbol{x}^{*}$.
\label{lem:uniqueRelative}
\end{lemma}

\begin{IEEEproof}
Since $u_{i}(x_{i})$ is strictly concave in $x_{i}$, then $u(\boldsymbol{x})=\sum_{i=1}^{I}u_{i}(x_{i})$ is strictly concave in $\boldsymbol{x}$. The feasible region defined by constraints (\ref{eqn:relative})-(\ref{eqn:positive}) is convex. Hence, $u(\boldsymbol{x})$ has a unique optimal solution $\boldsymbol{x}^{*}$ subject to constraints (\ref{eqn:relative})-(\ref{eqn:positive}). 
\end{IEEEproof}

{\subsection{Uniqueness of the socially optimal demand vector $\boldsymbol{q}^{*}$}}

Even though $u_{i}(\cdot)$'s are strictly concave in $x_{i}$, they are not strictly concave in the demand vector $\boldsymbol{q_{i}}$. Hence, SWO is non-strictly concave in $\boldsymbol{q}$. It is well-known that a non-strictly concave maximization problem might have several different global optimal optimizers (several different demand vectors $\boldsymbol{q}$ in our case) (see e.g. \cite{Lin:2006bc},\cite{Voice:2006it}). In particular, one can choose $c_{ij}$'s, $Q_{j}$'s, and $u_{i}(\cdot)$'s in such a way that a demand maximizing vector $\boldsymbol{q}^{*}$ of SWO is not unique. However, we can show that such cases arise with zero probability whenever channel offsets factors $c_{ij}$'s are independent random variables drawn from continuous distributions (see Assumption \ref{as:channelGainsOffsetFunction}). 

In the remainder of this section, we show that SWO has a unique maximizing demand vector with probability 1. We begin by proving Lemma \ref{lem:diffSupport}, which is an intermediate result stating that any two maximizing demand vectors of SWO must have different non-zero components. We then observe that any convex combination of two maximizing demand vectors is also a maximizing demand vector. Finally, we show that all convex combinations of maximizing demand vectors have the same non-zero components, which is a contradiction with Lemma \ref{lem:diffSupport}. {This} proves the main result of this section (Theorem \ref{thm:uniqueQ}).

To make our argument precise, we first define the support set of a demand vector $\boldsymbol{q}_{i}$ as follows. 

\begin{definition} (Support set): The support set $\hat{\mathcal{J}}_{i}(\boldsymbol{q}_{i})$ of a demand vector $\boldsymbol{q}_{i}$ contains the indices of its non-zero entries:
\begin{align*}
\hat{\mathcal{J}}_{i}(\boldsymbol{q}_{i})=\{j \in \mathcal{J}: q_{ij}>0\}.
\end{align*}
Given a demand vector $\boldsymbol{q}$, the ordered \emph{collection} of support sets $\hat{\mathcal{J}}_{1},\hat{\mathcal{J}}_{2},\ldots,\hat{\mathcal{J}}_{I}$ is denoted by $\{\hat{\mathcal{J}}_{i}\}_{i=1}^{I}$.
\label{def:support}
\end{definition}
The support set contains providers that user $i$ has strictly positive demand from.
\begin{lemma}
\label{lem:diffSupport}
Let $\boldsymbol{q}^{*}$ be an optimal solution of SWO (a maximizing demand vector) and $\{\hat{\mathcal{J}}_{i}\}_{i=1}^{I}$ be the corresponding collection of support sets. Then, $\boldsymbol{q}^{*}$ is almost surely\footnote{This holds on the probability space defined by the distributions of $c_{ij}$'s.} the unique maximizing demand vector corresponding to $\{\hat{\mathcal{J}}_{i}\}_{i=1}^{I}$.
\end{lemma}

\begin{IEEEproof}
For a maximizing demand vector $\boldsymbol{q}^{*}$, equations (\ref{eqn:relative})-(\ref{eqn:positive}) hold, with $\sum_{j=1}^{J}q^{*}_{ij}=x_{i}^{*}$. To prove the lemma, we will uniquely construct $\boldsymbol{q}^{*}$ from $\boldsymbol{x}^{*}$ and $\{\hat{\mathcal{J}}_{i}\}_{i=1}^{I}$. 

We can divide the users into two categories. The \emph{decided} users purchase from only one provider ($|\hat{\mathcal{J}}_{i}|=1$), and the \emph{undecided} users from several ($|\hat{\mathcal{J}}_{i}|>1$). It is also possible that some users have zero demand to all providers, but without loss of generality, we treat such users as decided. Recall that for all users we have $x^{*}_{i}=\sum_{j=1}^{J}q_{ij}^{*}c_{ij}$. For a \emph{decided} user $i$ who purchases only from provider $\bar{\jmath}$, this reduces to $x^{*}_{i}=q_{i\bar{\jmath}}^{*}c_{i \bar{\jmath}}$, and the corresponding unique demand vector is $\boldsymbol{q}_{i}^{*}=[0\cdots0 \frac{x_{i}^{*}}{c_{i\bar{\jmath}}}0\cdots 0]$.

For undecided users, finding {the unique} $\boldsymbol{q}_{i}^{*}$ is less straightforward as there is more than one $\boldsymbol{q}_{i}$ such that $\sum_{j \in \hat{\mathcal{J}_{i}}}^{}q_{ij}c_{ij}=x_{i}^{*}$. To show that the demand of undecided users is unique, we construct the bipartite graph representation (BGR) $\mathcal{G}$ of the undecided users' support sets as follows. We represent undecided users by circles, and providers of undecided users as squares. We place an edge $(i,j)$ between a provider node $j$ and a user node $i$ if $j \in  \hat{\mathcal{J}}_{i}$. 

We give an example of a BGR in Fig. \ref{fig:BGRgood}, where $\hat{\mathcal{J}}_{1}=\{a,b,c\}$, $\hat{\mathcal{J}}_{2}=\{b,d\}$, $\hat{\mathcal{J}}_{3}=\{d,e,f\}$ and $\hat{\mathcal{J}}_{4}=\{b,g\}$. 

\tikzstyle{user}=[circle,draw=black,thin, minimum size=2pt, transform shape] 
\tikzstyle{provider}=[rectangle,draw=black,fill=white, minimum size=5mm, transform shape]

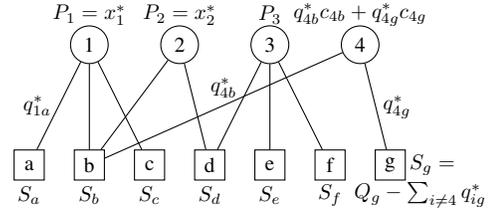
\begin{figure}[ht] 
   \centering
	\begin{tikzpicture}[>=stealth, scale=0.8]
		\node[transform shape] (user1) at (0,2) [user] {1};
		\node[transform shape] at (0,2.5){$P_{1}=x_{1}^{*}$};
		\node[transform shape] (user2) at (1.5,2) [user] {2};
		\node[transform shape] at (1.5,2.5){$P_{2}=x_{2}^{*}$};
		\node[transform shape] (user3) at (3,2) [user] {3};
		\node[transform shape] at (3,2.5){$P_{3}$};
		\node[transform shape] (user4) at (4.5,2) [user] {4};
		\node[transform shape] at (4.5,2.5){$q_{4b}^{*}c_{4b}+q_{4g}^{*}c_{4g}$};		
		\node[transform shape] (netA) at (-1,0) [provider] {a};
		\node[transform shape] at (-1,-0.5){$S_{a}$};		
		\node[transform shape] (netB) at (0,0) [provider] {b};
		\node[transform shape] at (0,-0.5){$S_{b}$};		
		\node[transform shape] (netC) at (1,0) [provider] {c};
		\node[transform shape] at (1,-0.5){$S_{c}$};		
		\node[transform shape] (netD) at (2,0) [provider] {d};
		\node[transform shape] at (2,-0.5){$S_{d}$};		
		\node[transform shape] (netE) at (3,0) [provider] {e};		
		\node[transform shape] at (3,-0.5){$S_{e}$};		
		\node[transform shape] (netF) at (4,0) [provider] {f};
		\node[transform shape] at (4,-0.5){$S_{f}$};		
		\node[transform shape] (netG) at (5,0) [provider] {g};
		\node[transform shape] at (5.7,0){$S_{g}=$};
		\node[transform shape] at (5.5,-0.5){$Q_{g}-\sum_{i \neq 4}^{}q_{ig}^{*}$};		
		\draw[transform shape] [-] (user1) -- node[left] {$q_{1a}^{*}$}(netA);
		\draw [-] (user1) -- (netB);
		\draw [-] (user1) -- (netC);
		\draw [-] (user2) -- (netB);
		\draw [-] (user2) -- (netD);
		\draw [-] (user3) -- (netD);
		\draw [-] (user3) -- (netE);
		\draw [-] (user3) -- (netF);
		\draw[transform shape] [-] (user4) -- node[above] {$q^{*}_{4b}$}(netB);
		\draw[transform shape] [-] (user4) -- node[right] {$q^{*}_{4g}$}(netG);
	\end{tikzpicture}
	\caption{Bipartite graph representation}
   	\label{fig:BGRgood}
\end{figure} 

{The BGR has the following properties\footnote{Fig. \ref{fig:BGRgood} shows a connected graph, but this need not be the case.} (see Fig. \ref{fig:BGRgood}):}

\begin{enumerate}
\item The sum of effective resource on all the edges connected to user $i$ is the optimal effective resource $x_{i}^{*}=\sum_{j \in  \hat{\mathcal{J}}_{i}}^{¥}q^{*}_{ij}c_{ij}=P_{i}$. Borrowing from coding theory and with some abuse of terminology, we call $P_{i}$ the \emph{check-sum} of user node $i$. \label{item:checkUser} 
\item The sum of all edges connected to provider node $j$ equals to the difference between the supply $Q_{j}$ and the demand from decided users who connect to provider $j$: $\sum_{i:(i,j)\in \mathcal{G}}^{¥}q^{*}_{ij}=Q_{j}-\sum_{i:(i,j)\notin \mathcal{G}}^{¥}q_{ij}^{*}=S_{j}$. {We call $S_{j}$ the \emph{check-sum} of provider node $j$.}
\label{item:checkNetwork}
\item With probability 1, the BGR does not contain any loops. {This is proved in Appendix \ref{app:diffSupport}.}
 \label{item:noLoops}
\end{enumerate}

{As it is the case in Fig. \ref{fig:BGRgood}}, the number of undecided users is smaller than the number of providers. This is a direct consequence of Property \ref{item:noLoops}), and will be proved later.

We can use the BGR to uniquely determine the demands of undecided users. Here we use Fig. \ref{fig:BGRgood} as an illustrative example, the formal description is given in Appendix \ref{app:algo}. We call this procedure the BGR decoding algorithm. Consider leaf node (a node with only one edge) $g$ and edge $q_{4g}^{*}$. The BGR implies that user $4$ is the only undecided customer of provider $g$. Since the demands of all decided users have been determined, then we know that $q_{4g}^{*}=S_{g}=Q_{g}-\sum_{i \neq 4}^{}q_{ig}^{*}$. We can then remove edge $q_{4g}^{*}$ and node $g$ from the BGR, and update the check-sum value of node $4$ to $P_{4}=x_{4}^{*}-q_{4g}^{*}c_{4g}$. Now consider node $4$ and edge $q_{4b}^{*}$. Since edge $q_{4b}^\ast$ is now the only edge connecting with user node 4, we have  $q_{4b}^{*}c_{4b}=P_{4}$ and hence $q_{4b}^{*}={P_{4}}/{c_{4b}}$. Next we can consider node $a$, $e$, or $f$, and so on. Property \ref{item:noLoops}) is crucial in this procedure since it guarantees that we can always find a leaf node in the reduced graph. 

In each step of the algorithm, we determine the unique value of $q_{ij}^{*}$ associated with the edge of one leaf. This  value is independent of the order in which we pick the leaf nodes, as seen from Appendix \ref{app:algo}. So, we can construct unique demand vector $\boldsymbol{q}_{i}^{*}$ for each undecided user $i$. Together with the unique demand vectors of the decided users, we have found the unique maximizing demand vector $\boldsymbol{x}^{*}$ of SWO with support sets $\{\hat{\mathcal{J}}_{i}\}_{i=1}^{I}$.
\end{IEEEproof}

\begin{theorem}
The social welfare optimization problem SWO has a unique maximizing solution $\boldsymbol{q}^\ast$ with probability 1. 
\label{thm:uniqueQ}
\end{theorem}

\begin{IEEEproof} 
The detailed proof is in Appendix \ref{app:proofUniqueDemand}, here we provide an outline. {Assume there exist two maximizing demand vectors of SWO which, by Lemma \ref{lem:diffSupport}, have different supports sets.} The support set of {a non-trivial} convex combination of {any two} non-negative vectors is the union of support sets of these two vectors. {Hence, all convex combinations of two maximizing demand vectors of SWO, which are also maximizing demand vectors, have the same support. This is a contradiction to Lemma \ref{lem:diffSupport}.}
\end{IEEEproof}
 
Given an optimal demand vector $\boldsymbol{q}^{*}$ of the SWO problem, there exists a unique corresponding Lagrangue multiplier vector $\boldsymbol{p}^\ast$, associated with the resource constraints of $J$ providers \cite{Bertsekas:1999yq}. Next, we show that $(\boldsymbol{q}^\ast,\boldsymbol{p}^\ast)$ is the unique equilibrium of the provider competition game defined in Section \ref{sec:model}.  
 
\section{Analysis of the two-stage game}

\label{sec:oligopoly} 
In this section we show that there exists a unique equilibrium (defined more precisely shortly) of the {multi-leader-follower} provider competition game. In particular, this equilibrium corresponds to the unique social optimal solution of SWO and the associated Lagrange multipliers. The idea is to interpret the Lagrange multipliers as the prices announced by the providers. Moreover, we show that there are at most $J-1$ undecided users at this equilibrium. 

First, we define the equilibrium concept \cite{Fudenberg:1991}:

\begin{definition} (Subgame perfect equilibrium (SPE)) A price demand tuple $\left(\boldsymbol{p}^{*}, \boldsymbol{q}^{*}(\boldsymbol{p}^{*})\right)$ is a subgame perfect equilibrium for the provider competition game if no player has an incentive to deviate unilaterally {at any stage of the game}. In particular, each user $i\in\mathcal{I}$ maximizes its payoff given prices $\boldsymbol{p}^\ast$. Each provider $j\in\mathcal{J}$ maximizes its revenue given other providers' prices $p^{*}_{-j}=(p^{*}_{1},\cdots,p^{*}_{j-1}, p^{*}_{j+1},\cdots p^{*}_{J})$ and the users' demand $\boldsymbol{q}^{*}(\boldsymbol{p}^{*})$.
\label{def:SPE}
\end{definition}
 
We will compute the equilibrium concept using backward induction. {In Stage II, we will compute the best response of the users $\boldsymbol{q}^{*}(\boldsymbol{p})$ as a function of any given price vector $\boldsymbol{p}$. Then in Stage I, we will compute the equilibrium prices $\boldsymbol{p}^{*}$. For equilibrium prices $\boldsymbol{p}^{*}$, the best response of the users $\boldsymbol{q}^{*}(\boldsymbol{p}^{*})$ is uniquely determined via BGR decoding.}

\subsection{Equilibrium strategy of the users in Stage II}
Consider users facing prices $\boldsymbol{p}$ in the second stage. Each user solves a user payoff maximization (UPM) problem:
\begin{align}
\label{eqn:defUPM}
\textbf{UPM}: \max_{\boldsymbol{q}_{i}\geq \boldsymbol{0}} v_{i}= \max_{\boldsymbol{q}_{i}\geq \boldsymbol{0}} u_{i}\left( \sum_{j=1}^{J}{q_{ij}}{c_{ij}}\right)-\sum_{j=1}^{J}p_{j}q_{ij}
\end{align}
\begin{lemma}
\label{lem:uniqueEffectiveUPM}  
For each user $i\in \mathcal{I}$, there exists a unique nonnegative value $x_{i}^{*}$, such that $\sum_{j=1}^{¥}c_{ij}q_{ij}=x_{i}^{*}$ for every maximizer $\boldsymbol{q}_i$  of the UPM problem. Furthermore, for any $j$ such that $q_{ij}>0$, $\frac{p_{j}}{c_{ij}}=\min_{k\in \mathcal{J}}\frac{p_{k}}{c_{ik}}$. 
\end{lemma}

Proof is given in Appendix \ref{app:proofUPM}. We remark that, together, the unique $x_{i}^{*}$'s from Lemma \ref{lem:uniqueEffectiveUPM} form a vector $\boldsymbol{x}^{*}$, which is equal to $\boldsymbol{x}^{*}$, the SWO maximizer from Lemma \ref{lem:uniqueRelative}.

\begin{definition}(Preference set)
For any price vector $\boldsymbol{p}$, user $i$'s preference set $\mathcal{J}_{i}(\boldsymbol{p})$ includes each provider $j \in \mathcal{J}$ with $\frac{p_{j}}{c_{ij}}=\min_{k \in \mathcal{J}} \frac{p_{k}}{c_{ik}}$.
\end{definition}
 In light of Lemma \ref{lem:uniqueEffectiveUPM}, $\mathcal{J}_{i}$ is the set of providers from which user $i$ might request a strictly positive amount of resource. Users can again be partitioned to decided and undecided based on the cardinality of their preference sets, analogous to the distinction made in Section \ref{sec:socOpt}. The preference set of a decided user $i$ contains a singleton, and there is a unique vector $\boldsymbol{q}_{i}$ that maximizes his payoff. By contrast, the preference set of an undecided user $i$ contains more than one provider, and any choice of $\boldsymbol{q}_{i}\geq \boldsymbol{0}$ such that $x^{*}_{i}=\sum_{j \in \mathcal{J}_{i}}^{¥}q_{ij}c_{ij}$ maximizes his payoff. 

There is a close relationship between the support sets from Section \ref{sec:socOpt} and preference sets defined here. {Facing prices $\boldsymbol{p}$, a user $i$ \emph{may} request positive resource only from providers who are in his preference set $\mathcal{J}_{i}$. By definition, he \emph{actually} {requests} positive resource from providers who are in his support set $\hat{\mathcal{J}}_{i}$.} So the support set of a user is a subset of his preference set: $\hat{\mathcal{J}}_{i}(\boldsymbol{q(\boldsymbol{p})})\subset \mathcal{J}_{i}(\boldsymbol{p})$. We can construct a BGR based on the preference sets and show that this BGR also has no loops with probability 1, following a similar proof (Appendix \ref{app:diffSupport}).

Suppose that the optimal Lagrange multipliers $\boldsymbol{p}^{*}$ from Section \ref{sec:socOpt} are announced as prices. Since all users have access to complete network information, each of them can calculate all users' preference sets, and can construct the corresponding BGR. Undecided users can now uniquely determine their demand vector by independently running the same BGR decoding algorithm. The demand found through BGR decoding is unique as all demand vectors are considered at one time and equality of supply and demand is taken into account. We note that the demand found in this way is only one of an undecided user's infinitely many best responses under prices $\boldsymbol{p}^{*}$. However, only the demands given by the BGR decoding algorithm will balance the supply and demand for each provider at the optimal price $\boldsymbol{p}^{*}$. We will later show that this is the only subgame perfect equilibrium of the provider competition game.

\subsection{Equilibrium strategy of the providers in Stage I}
The optimal choice of prices for the providers depends on how the users' demand changes with respect to the price, which further depends on the users' utility functions. The quantity that indicates how a user's demand changes with respect to the price is the \emph{coefficient of relative risk aversion} \cite{Mas-Colell:1995} of utility function $u_{i}$, i.e. $k^{i}_{RRA}=-\frac{xu_{i}''(x)}{u_{i}'(x)}$.  We focus on a class of utility functions characterized in Assumption \ref{as:RRA}. 
\begin{assumption}
\label{as:RRA}
For each user $i\in\mathcal{I}$, the coefficient of relative risk aversion of its utility function is less than 1.
\end{assumption}

Assumption \ref{as:RRA} is satisfied by some commonly used utility functions, such as $\log(1+x)$ and the $\alpha-$fair utility functions $\frac{x^{1-\alpha}}{1-\alpha}$, for $\alpha \in (0, 1)$ \cite{Rad:2009sf}. Under Assumption \ref{as:RRA}, a monopolistic provider will sell all of its resource $Q_j$ to maximize its revenue. Intuitively, when a provider lowers its price, the demand of the users increases significantly enough that the change in revenue of the provider is positive. This encourages the provider to lower the price further such that eventually total demand equals total supply. In the case of multiple providers, Assumption \ref{as:RRA} also ensures that all providers are willing to sell all their resources to maximize their revenues.

{We call the prices that achieve equality of demand and supply \emph{market clearing prices}.}

\begin{theorem}
\label{thm:uniqueNash}
Under Assumption \ref{as:RRA}, the unique socially optimal demand vector $\boldsymbol{q}^{*}$ and the associated Lagrangian multiplier vector $\boldsymbol{p}^{*}$ of the SWO problem constitute the unique sub-game perfect equilibrium of the provider competition game. 
\end{theorem}

The proof is given in Appendix \ref{app:proofSPE}. {It is interesting to see that the competition of providers does not reduce social efficiency. This is not a simple consequence of the strict concavity of the users' utility functions; it is also related to the elasticity of users' demands. Assumption \ref{as:RRA} ensures that the demands are elastic enough such that a small decrease in price leads to significant increase in demand and thus a net increase in revenue. }

Under the optimal prices $\boldsymbol{p}^{*}$  announced by the providers in the first stage, the users in the second stage will determine the unique demand vector $\boldsymbol{q}^{*}$ using BGR decoding. {On the other hand, if the providers charge prices other than $\boldsymbol{p}^{*}$, no best-response from the users will make the demand equals to the supply, which is a necessary condition for an equilibrium.}

{In light of Theorem \ref{thm:uniqueNash}, we will refer to the unique subgame perfect equilibrium $(\boldsymbol{p}^{*},\boldsymbol{q}^{*})$ of the provider competition game as \textbf{the equlibrium}.} 
\subsection{The number of undecided users} 

Since the presence of undecided users makes the analysis challenging, it is interesting to understand how many undecided users there can be in a given game. It turns out that such number is upperbounded by the number of providers $J$ in the network. 

\begin{lemma}
Under any given price vector $\boldsymbol{p}$ in the first stage, the number of undecided users in the second stage is strictly less than $J$. 
\label{lem:nbrUndecided}
\end{lemma} 

The proof is given in Appendix \ref{app:proofNumberUndecided}. {The main idea is that if the number of undecided user nodes in a BGR is not smaller than the number of provider nodes, then there exists a loop in the BGR. This, however, occurs with zero probability, as shown in Section \ref{sec:socOpt}}.  

\section{Primal-Dual Algorithm}
\label{sec:primal-dual}

The previous analysis of the subgame perfect equilibrium has assumed that every player (provider or user) knows the complete information of the system. This may not be true in practice. In this section we present a distributed primal-dual algorithm where providers and users only know local information and make local decisions in an iterative fashion. We show that such algorithm globally converges to the unique equilibrium discussed in Theorem \ref{thm:uniqueNash} under mild conditions on the updating rates. 

The key proof idea is to show that the primal-dual algorithm converges to a set containing the optimal solution of SWO. We can further show that this set contains only the unique optimal solution in most cases, regardless of the values of the updating rates. 

We first present the algorithm, and then the proof of its convergence. 

\subsection{Primal-dual algorithm}
In this section we will consider a continuous-time algorithm, where all the variables are functions of time. For compactness of exposition, we will sometimes write $q_{ij}$ and $p_{j}$ when we mean $q_{ij}(t)$ and $p_{j}(t)$, respectively. Their time derivatives $\frac{\partial q_{ij}}{\partial t}$ and $\frac{\partial p_{j}}{\partial t}$ will often be denoted by $\dot{q}_{ij}$ and $\dot{p}_{j}$. {We denote by $\boldsymbol{q}^{*}$ and $\boldsymbol{p}^{*}$ the unique maximizer of SWO and the corresponding Lagrange multiplier vector, respectively. As shown in Theorem \ref{thm:uniqueNash}, $(\boldsymbol{p}^{*}, \boldsymbol{q}^{*})$ is also the unique subgame perfect equilibrium of the provider competition game. These values are constant.}

To simplify the notation, we denote by $f_{ij}(t)$ or simply $f_{ij}$ the marginal utility of user $i$ with respect to $q_{ij}$ when his demand vector is $\boldsymbol{q_{i}}(t)$:
\begin{align}
f_{ij}=\frac{\partial u_{i}(\boldsymbol{q_{i}})}{\partial q_{ij}}=c_{ij} \frac{\partial u_{i}(x)}{\partial x}\Big\vert_{x=x_{i}=\sum_{j=1}^{J}q_{ij} c_{ij}} \label{eqn:marginalUtility}.
\end{align}
We will use $f_{ij}^{*}$ to denote the value of $f_{ij}(t)$ evaluated at $\boldsymbol{q}_{i}^{*}$, the maximizing demand vector of user $i$. So, $f_{ij}^{*}$ is a constant that is equal to a user's equilibrium marginal utility as opposed to $f_{ij}(t)$ which indicates marginal utility at a particular time $t$. We also define $\nabla u_{i}(\boldsymbol{q_{i}})=[f_{i1} \cdots f_{iJ}]^{T}$ and $\nabla u_{i}(\boldsymbol{q_{i}^{*}})=[f_{i1}^{*} \cdots f_{iJ}^{*}]^{T}$, where all the vectors are column vectors.

We define $(x)^{+}=\max(0,x)$ and 
\begin{align*}
(x)_{y}^{+}=\begin{cases}
x & y >0  \\
(x)^{+} & y\leq0.
\end{cases}
\end{align*}
Another way to think of this notation is $(x)_{y}^{+}=x(1-\mathbbm{1}_{(-\infty,0]}(x)\mathbbm{1}_{(-\infty,0]}(y)$), where $\mathbbm{1}$ is the indicator function{, i.e., $\mathbbm{1}_{A}(x)=1$ if $x\in A$, and $0$ otherwise}. 

Motivated by the work in \cite{Chen:2008la}, we consider the following standard \textbf{primal-dual variable update algorithm}:
\begin{table}[htbp]
\begin{align}
 & \dot{q}_{ij}=k_{ij}^{q}\left( f_{ij}-p_{j} \right)_{q_{ij}}^{+},\; \label{eqn:primal} i \in \mathcal{I},j\in \mathcal{J}\\
 & \dot{p}_{j}=k_{j}^{p}\left( \sum_{i=1}^{I}q_{ij}-Q_{j} \right)_{p_{j}}^{+},\;\label{eqn:dual}j\in \mathcal{J}.
\end{align}
\label{tab:primalDual}
\end{table}

\noindent Here $k_{ij}^{p}$, $k_{j}^{p}$ are the constants representing update rates. The update rule ensures that, when a variables of interest ($q_{ij} $ or $  p_{j}$) is already zero, it will not become negative even when the direction of the update (i.e. quantity in the parenthesis) is negative. The tuple $(\boldsymbol{q}(t),\boldsymbol{p}(t))$ controlled by equations (\ref{eqn:primal}) and (\ref{eqn:dual}) will be referred to as the \emph{solution trajectory} of the differential equations system defined by (\ref{eqn:primal}) and (\ref{eqn:dual}). 

{The motivation for the proposed algorithm is quite natural. A provider increases its price when the demand is higher than its supply and decreases its price when the demand is lower. A user decreases his demand when a price is higher than his marginal utility and increases it when a price is lower.} In essence, the algorithm is following the natural direction of market forces.

One key observation is that these updates can be implemented in a distributed fashion. The users only need to know the prices proposed by the providers. The providers only need to know the demand of the users for their own resource, and not for the resource of other providers (as was the case in the analysis of Section \ref{sec:oligopoly}). In particular, only user $i$ needs to know his own channel offset parameters $c_{ij}$,  $j \in \mathcal{J}$.

The first step to prove the algorithm's convergence is to construct a lower-bounded La Selle function $V(\boldsymbol{q}(t),\boldsymbol{p}(t))$ and show that its value is non-increasing for any solution trajectory $(\boldsymbol{q}(t),\boldsymbol{p}(t))$ that satisfies \eqref{eqn:primal} and \eqref{eqn:dual}. {This will ensure that $(\boldsymbol{q}(t),\boldsymbol{p}(t))$ converge to a set of values that keeps ${V}(\boldsymbol{q}(t),\boldsymbol{p}(t))$ constant.}
 
\subsection{Convergence of the primal-dual to an invariant set}

We consider the following La Salle function:

\begin{align*}
&V(\boldsymbol{q}(t),\boldsymbol{p}(t))=V(t) \\
&=\sum_{i,j}^{¥}\frac{1}{k_{ij}^{q}}\int_{0}^{q_{ij}(t)}(\beta-q_{ij}^{*})d \beta+\sum_{j}^{¥}\frac{1}{k_{j}^{p}}\int_{0}^{p_{j}(t)}(\beta-p_{j}^{*})d \beta,
\end{align*}

{It can be shown} that $V(\boldsymbol{q}(t),\boldsymbol{p}(t))\geq V(\boldsymbol{q}^{*},\boldsymbol{p}^{*})$, i.e., $V$ is bounded from below. {This ensures that if the function $V$ is non-increasing, it will eventually reach a constant value (which may or may not be the global minimum $V(\boldsymbol{q}^{*},\boldsymbol{p}^{*})$).}

The derivative of $V$ \emph{along the solution trajectories of the system}, $\frac{\partial V}{\partial t}$, denoted by $\dot{V}$, is given by:
\begin{align*}
\dot{V}(t)=\sum_{i,j}^{¥}\frac{\partial V}{\partial q_{ij}} \dot{q}_{ij}+\sum_{j}^{¥}\frac{\partial V}{\partial p_{j}} \dot{p}_{j}.
\end{align*}

\begin{lemma}
\label{lem:decreasingLaSalleFunction}
The value of the La Selle function $V$ is non-increasing along the solution trajectory, defined by (\ref{eqn:primal}) and (\ref{eqn:dual}), i.e. $\dot{V}(t) \leq 0$. 
\end{lemma}

\begin{proof}
Proof is given in Appendix \ref{app:proofVdotLessThanZero}. The proof manipulates the expression for $\dot{V}$ and shows that it can be reduced to the following form:
\begin{align}
\dot{V}\leq \sum_{i}^{¥}\left(\sum_{j}(q_{ij}(t)-q_{ij}^{*})(f_{ij}(t)-f_{ij}^{*})\right)\nonumber \\+\sum_{i,j}^{¥}(q_{ij}(t)-q_{ij}^{*})(f_{ij}^{*}-p_{j}^{*}).
\label{eqn:VdotLessThenZero}
\end{align}
Using concavity of $u_{i}'s$ and properties of the equilibrium point $\boldsymbol{q}^{*},\boldsymbol{p}^{*}$, we can show that individual elements of the summations in (\ref{eqn:VdotLessThenZero}) are non-positive.
\end{proof} 

Combining Lemma \ref{lem:decreasingLaSalleFunction} and the  La Salle's \emph{invariance principle} (Theorem 4.4 of \cite{Khalil:2002uq}, p. 128) we can prove the following:

\begin{proposition}
The pair $\boldsymbol{q}(t),\boldsymbol{p}(t)$ converges to the invariant set $V_{L}=\{\boldsymbol{q}(t),\boldsymbol{p}(t): \dot{V}(\boldsymbol{q}(t),\boldsymbol{p}(t))=0\}$ as $t \rightarrow \infty$.
\label{prop:covergenceToVL}
\end{proposition}

{It is clear that the invariant set $V_{L}$ contains the solution trajectory that has the value of the unique maximizer of SWO $(\boldsymbol{q}^{*}(t),\boldsymbol{p}^{*}(t))=(\boldsymbol{q}^{*},\boldsymbol{p}^{*})$ for all $t$, since $\dot{V}(\boldsymbol{q}^{*},\boldsymbol{p}^{*})=0$. However it may contain other points as well. When the trajectory $(\boldsymbol{q}(t),\boldsymbol{p}(t))$ enters the invariant set, it either reaches its minimum (i.e., by converging to the unique equilibrium point $(\boldsymbol{q}^{*},\boldsymbol{p}^{*})$), or it gets stuck permanently in some limit cycle.
In either case, the trajectory will be confined to a subset of $V_{L}=\{(\boldsymbol{q}(t),\boldsymbol{p}(t)): \dot{V}(\boldsymbol{q}(t),\boldsymbol{p}(t))=0\}$.}

{The remainder of this section is to show that the invariant set $V_{L}$ contains only the equilibrium point $(\boldsymbol{q}^{*},\boldsymbol{p}^{*})$. This will be done in two steps. First, we show that the set $V_{L}$ has only one element for the majority of provider competition instances, without any restrictions on the variable update rates. Second, we provide a sufficient condition on the update rates so that the global convergences to the unique equilibrium point is also guaranteed in the remaining instances. } 

\subsection{Convergence when providers have decided customers}
\label{subsec:convergeFixed}
{In the following two sections we consider the properties of the solution trajectory \textbf{on the invariant set} $V_{L}$.}

The proof of Lemma \ref{lem:decreasingLaSalleFunction} shows that individual terms on the right-hand side of (\ref{eqn:VdotLessThenZero}) are non-positive. Combined with Proposition \ref{prop:covergenceToVL}, we get the following result:
\begin{corollary}
On the invariant set $V_{L}$, $\boldsymbol{q}(t), \boldsymbol{p}(t)$ are such that:
\begin{align}
\sum_{j} & (q_{ij}(t)-q_{ij}^{*})(f_{ij}(t)-f_{ij}^{*}) \nonumber \\ = & (\nabla u_{i}(\boldsymbol{q}_{i})(t)-\nabla u_{i}(\boldsymbol{q}_{i}^{*}))^{T}(\boldsymbol{q}_{i}(t)-\boldsymbol{q}_{i}^{*})  = 0, \forall i\in \mathcal{I} \label{eqn:sale1}\\
& (q_{ij}(t)-q_{ij}^{*})(f_{ij}^{*}-p_{j}^{*}) = 0, \text{ for all } i \in \mathcal{I},\; j\in \mathcal{J}, \label{eqn:sale2}
\end{align}
\end{corollary}
where we recall that $\nabla u_{i}(\boldsymbol{q_{i}})(t)=[f_{i1} \cdots f_{iJ}]^{T}$ and $\nabla u_{i}(\boldsymbol{q_{i}^{*}})=[f_{i1}^{*} \cdots f_{iJ}^{*}]^{T}$.

{Expressions (\ref{eqn:sale1}) and (\ref{eqn:sale2}) give basic properties of the solution trajectories $\boldsymbol{q}(t),\boldsymbol{p}(t)$ on the invariant set. We next prove two intermediate results that give further characterization of $\boldsymbol{q}(t),\boldsymbol{p}(t)$ on $V_{L}$.}

\begin{lemma}
{For any point in the invariant set $V_L$, we have $f_{ij}(t)=f_{ij}^{*}$ for all $i\in \mathcal{I}, j\in \mathcal{J}$. In other words, any user $i$'s marginal utility with respect to its demand of any provider $j$ equals to the corresponding value at the unique equilibrium. In addition, $ q_{ij}(t)(f_{ij}^{*}-p_{j}^{*})=0$ on the invariant set $V_{L}$.}\label{lem:convergesToEquilMarginUtil}
\end{lemma}

The proof is given in Appendix \ref{app:MarginalUtilEqualEquilUtil}. 
{An equivalent way of stating the first part of Lemma \ref{lem:convergesToEquilMarginUtil} is that $x_{i}(t)=x^{*}_{i}$ on $V_{L}$: 
\begin{align}
\sum_{j}^{¥}q_{ij}(t)c_{ij}=\sum_{j}^{¥}q_{ij}^{*}c_{ij}, \text{ for all } i \in \mathcal{I} \label{eqn:sale3}.
\end{align}}
\noindent{The second part of Lemma \ref{lem:convergesToEquilMarginUtil} claims that $ q_{ij}(t)(f_{ij}^{*}-p_{j}^{*})=0$. From the proof of Lemma \ref{lem:uniqueEffectiveUPM}, we know $q^{*}_{ij}(f_{ij}^{*}-p_{j}^{*})=0$. So for $i,j$ such that $f_{ij}^{*}<p_{j}^{*}$, $q_{ij}^{*}=0$ implies that $q_{ij}(t)=0$. This is good news, since from Lemma \ref{lem:nbrUndecided} most users have zero demand to all but one provider at the unique equilibrium. Now we know that the same holds on the invariant set $V_{L}$. Similarly, $q_{ij}(t)>0$ only if $f^{*}_{ij}=p_{j}^{*}$.}

Lemma \ref{lem:convergesToEquilMarginUtil} does not preclude the possibility that a demand $q_{ij}(t)$ for a user with $f^{*}_{ij}=p_{j}^{*}$ may oscillate between being zero and being strictly positive. The following Lemma shows that this is not possible: 

\begin{lemma} 
The set $\hat{\mathcal{J}}_{i}(t)=\{j \in \mathcal{J}: q_{ij}(t)>0\}$ does not change over time on the invariant set. In addition, $p_{j}(t)>0$ on the invariant set for all $j\in \mathcal{J}$. 
\label{lem:borderq}
\end{lemma}
{The proof is given in Appendix \ref{app:supportConstantOnInvariant}. Lemma \ref{lem:borderq} implies that if $q_{ij}(t)=0$ on the invariant set $V_{L}$, then $\dot{q}_{ij}(t)=0$ on $V_{L}$. Also, if $q_{ij}(t)>0$ then $\dot{q}_{ij}(t)=k_{i j}^{q}\left(f_{i j}^{*}-p_{j}\right)$, as strictly positive demand stays strictly positive on the invariant set.}

We are now ready to claim the main result of this subsection:
\begin{theorem}
{A demand vector of a decided user $i\in \mathcal{I}$ converges to the equilibrium demand vector, i.e. $\lim_{t\rightarrow \infty}\boldsymbol{q}_{i}(t)=\boldsymbol{q}^{*}_{i}$. The price $p_{j}$ of any provider $j$ who has at least one decided user at the equilibrium $(\boldsymbol{q}^{*},\boldsymbol{p}^{*})$, converges to the equilibrium price, i.e. $\lim_{t\rightarrow \infty} p_j(t) =p_j^\ast$. }
\label{thm:fixedProviders}
\end{theorem}
\begin{proof} {Consider an arbitrary user $i$ who at the equilibrium has only one preferred provider $\bar{\jmath}$, i.e., $f_{i\bar{\jmath}}^{*}=p_{\bar{\jmath}}^{*}$ for $\bar{\jmath}$ such that $q^{*}_{i\bar{\jmath}}>0$, and $f_{ij}^{*}<p_{j}^{*}$ for $j \neq \bar{\jmath}$. By Lemma \ref{lem:convergesToEquilMarginUtil}, this implies $q_{i\bar{\jmath}}(t)>0$, $q_{ij}(t)=0$  and $\dot{q}_{ij}(t)=0$ for all $j\neq\bar{\jmath}$. Combined with (\ref{eqn:sale3}), this means that $q_{i\bar{\jmath}}(t)=q_{i\bar{\jmath}}^{*}$ and $\dot{q}_{i\bar{\jmath}}(t)=0$, i.e., user $i$'s demand vector converges to its equilibrium value.} For provider $\bar{\jmath}$, $q_{i\bar{\jmath}}>0$ so by Lemma \ref{lem:borderq}:
\begin{align*}
0=& \dot{q}_{i\bar{\jmath}}(t)=k_{i\bar{\jmath}}^{q}\left(f_{i\bar{\jmath}}^{*}-p_{\bar{\jmath}}(t)\right)^{+}_{q_{i\bar{\jmath}}}=k_{i\bar{\jmath}}^{q}\left(f_{i\bar{\jmath}}^{*}-p_{\bar{\jmath}}(t)\right).
\end{align*}

From this it follows that $p_{\bar{\jmath}}(t)=f_{i\bar{\jmath}}^{*}=p_{\bar{\jmath}}^{*}$. Further differentiation yields $\dot{p}_{\bar{\jmath}}(t)=0$, meaning that the prices have also converged, which completes the proof.
\end{proof}
{\begin{theorem}
If every provider has at least one decided customer in the unique equilibrium of the provider competition game, the primal dual algorithm converges to this equilibrium.
\end{theorem}
{\begin{proof}
By Theorem~\ref{thm:fixedProviders}, the prices of all providers and the demand vectors of all decided users converge on the invariant set. It remains to be shown that the demand vectors of undecided users also converge. By an argument similar to the proof of Lemma~\ref{lem:diffSupport}, we can draw a BGR for the undecided users. The demand vectors of undecided users that satisfy the constraints on the BGR are unique, and thus also converge.
\end{proof}}}

{In most practical cases, where the number of users is much larger than the number of providers, all providers will have at least one decided user, and hence the convergence of the primal-dual algorithm is guaranteed. Next we study the more complicated case where some providers do not have any associated decided users. In that case, we can still prove the convergence of the primal-dual algorithm under mild conditions of the variable update rates. }

\subsection{Convergence when providers have no decided customers}
\label{subsec:convergeLoose}
Without loss of generality, we now focus on the problem where all the users are undecided at the equilibrium. If we can prove that the algorithm converges in this case, then we can also prove convergence in the more general case where some providers have decided users. Let $I$ be the number of undecided users, and $J$ the number of providers. From Lemma \ref{lem:nbrUndecided}  we know that $I<J$. 

\begin{theorem}
\label{thm:looseProviders}
{Let $I<J$, and suppose that at the equilibrium $|\{j \in \mathcal{J}:q_{ij}^{*}>0\}|>1$ for all $i \in \mathcal{I}$. The primal-dual algorithm converges to the unique equilibrium if the price update rates $k^{p}_{j}$ are not integer multiples of each other and the demand update rates are equal, i.e. $k^{q}_{ij}=k$ $\forall i \in \mathcal{I}, \forall j\in \mathcal{J}$.}
\end{theorem}

\begin{IEEEproof} {We first define some matrices to facilitate the proof. Let $C$ be the $I \times IJ$ matrix of channel offsets $c_{ij}$. Let $\mathbb{I}_{J}$ be the identity matrix of size $J \times J$. Define matrix $A$ to be the $IJ \times J$ matrix of $I$ vertically stacked identity matrices. Let $K^{p}=diag(k^{p}_{j}, j \in \mathcal{J})$ be a $J\times J$ diagonal matrix containing price update rates. Let the $IJ\times IJ$ diagonal matrix $K^{q}$ be the matrix of demand update rates whose $y^{th}$ entry is $K^{q}_{y,y}=k^{q}_{i,j}$, where $i= \lfloor (y-1)/J \rfloor+1 $ and $j=((y-1) \mod J)+1$, and the off-diagonal entries are all zeros. Here $\lfloor x \rfloor$ is the floor function representing the largest integer not larger than $x$. }

By (\ref{eqn:sale3}), we know that the effective resource of all users has converged on the invariant set $V_{L}$. We rewrite this as $\boldsymbol{c}^{T}_{i}(t)\boldsymbol{q}_{i}=\boldsymbol{c}^{T}_{i}\boldsymbol{q}^{*}_{i}$ for all $i \in \mathcal{I}$, or $C \boldsymbol{q}(t)=C\boldsymbol{q}^{*}$.

We want to express the primal-dual update algorithm (\ref{eqn:primal}) and (\ref{eqn:dual}) in matrix form. The final hurdle is getting rid of the $(x)^{+}_{y}$ operation. From Lemma \ref{lem:borderq} we know that $\boldsymbol{p}(t)>0$ and that the support sets of vectors $\boldsymbol{q}(t)$ do not change on the invariant set. Hence, we can write $\dot{q}_{ij}=k_{ij}^{q}\left( f_{ij}-p_{j} \right)_{q_{ij}}^{+}= k_{ij}^{q}\left( f^{*}_{ij}-p_{j} \right)\mathbbm{1}_{(0,\infty)}(q_{ij})$ (recall that $f_{ij}=f_{ij}^{*}$ on the invariant set). This enables us to revise the definition of the update rates matrix to be $\hat{K}^{q}=diag(k^{q}_{ij}\mathbbm{1}_{(0,\infty)}(q_{ij}), i \in \mathcal{I}, j \in \mathcal{J})$. 

Then, (\ref{eqn:primal}) and (\ref{eqn:dual}) can be written as:
\begin{align}
\dot {\boldsymbol{q}}= &\hat{K}^{q} \left(\boldsymbol{f}^{*}-A\boldsymbol{p}(t)\right) \label{eqn:derivMatrix1}\\
\dot {\boldsymbol{p}}= &K^{p}\left(A^{T}\boldsymbol{q}(t)-Q\right).\label{eqn:derivMatrix2}
\end{align}
{Notice that \eqref{eqn:derivMatrix1} and \eqref{eqn:derivMatrix2} form a system of linear equations, so the non-linear primal dual dynamics defined in (\ref{eqn:primal}) and (\ref{eqn:dual}) becomes linear on the invariant set.} {The following result paves the way to showing that $\boldsymbol{p}(t)$ is constant on $V_{L}$.}

\begin{lemma}
\label{lem:BpequalsE}
{Let $E$ and $\mathcal{B}$ be constant matrices, where $\mathcal{B}=\left[B ; BD ; \cdots; BD^{J-1}\right]$ ($\mathcal{B}$ is a tall matrix), $B=C\hat{K}^{q} A$, and $D=K^{p}A^{T}\hat{K}^{q} A$. The dimensions of $\mathcal{B}$, $B$ and $D$ are $IJ\times J$, $I\times J$ and $J\times J$, respectively. On the invariant set, $\mathcal{B}\boldsymbol{p}(t)=E$.}
\end{lemma}
\begin{IEEEproof}
The proof is obtained by repeatedly differentiating Equations (\ref{eqn:derivMatrix1}) and (\ref{eqn:derivMatrix2}) with respect to time. The more detailed calculation is given in Appendix \ref{app:stackMatricesPrice}. 
\end{IEEEproof}

{If we can prove that the rank of $\mathcal{B}$ is $J$, then we could write $\mathcal{B}_{J}\boldsymbol{p}(t)=E$, where $\mathcal{B}_{J}$ is a $J\times J$ matrix constructed by taking $J$ linearly independent rows of $\mathcal{B}$. Then $\boldsymbol{p}(t)=E\mathcal{B}_{J}^{-1}$, which implies that $\boldsymbol{p}(t)$ converges on the invariant set.} To show that the rank of matrix $\mathcal{B}$ is $J$, we use Theorem 6.01 from \cite{Chen:1998uq} (its proof is similar to that of Theorem 6.1 in \cite{Chen:1998uq}):
\begin{theorem}
A matrix $\mathcal{B}=\left[B ; BD ; \cdots; BD^{J-1}\right]$ has full column rank if and only if matrix $G=\left[\begin{array}{c}B \\D-\lambda_{j} \mathbb{I} \end{array}\right]$ has full column rank for all eigenvalues $\lambda_{j}$, $j\in \mathcal{J}$ of $D$.
\label{thm:matrixEquiv}
\end{theorem}

We now provide a sufficient condition to ensure convergence of the primal-dual algorithm. 

\begin{lemma}
{Let matrices $B$ and $D$ be defined as in Lemma \ref{lem:BpequalsE}.} Matrix $G=\left[\begin{array}{c}B \\D-\lambda \mathbb{I} \end{array}\right]$ has full column rank for all eigenvalues $\lambda$ of $D$ if $k^{p}_{j}\neq ak^{p}_{j'}\;  \forall j,j' \in \mathcal{J}$ and all $a \in \mathbb{N}^{+}$ (i.e., as long as the price update rates are not integer multiples of each other) and $k^{q}_{ij}=k$, for all $i \in \mathcal{I}, j\in \mathcal{J}$. 
\label{lem:rankJ}
\end{lemma}

\begin{IEEEproof}
The proof is given in Appendix \ref{app:proofMatrixMinusEVfullRank}.
\end{IEEEproof}

Combining Lemma \ref{thm:matrixEquiv} and Lemma \ref{lem:rankJ}, we see that a unique vector $\boldsymbol{p}(t)$ can be computed from equation $\mathcal{B}\boldsymbol{p}(t)=E$, meaning that $\boldsymbol{p}(t)$ takes a single value on the invariant set and does not change with time. This also means that the demand vector $\boldsymbol{q}(t)$ does not change on the invariant set. {So since $\dot{\boldsymbol{p}}=0$ and $\dot{\boldsymbol{q}}=0$, the primal dual algorithm converges to an equilibrium point, which we call a dynamic equilibrium point}. It can be shown that the dynamic equilibrium point is constrained by the same set of equations as the unique equilibrium of the provider competition game (refer to the proof of Theorem \ref{thm:uniqueNash}). {Hence, there is only one element in the invariant set $V_{L}$, and it corresponds to the equilibrium of the provider competition game, $(\boldsymbol{q}^{*},\boldsymbol{p}^{*})$.} This concludes the proof of Theorem~\ref{thm:looseProviders}.
\end{IEEEproof}

Note that that the condition on the update rates in Theorem \ref{thm:looseProviders} is sufficient but not necessary. In fact, by looking at the form of the $D$ matrix from the proof of Lemma \ref{lem:rankJ}, we can see that a sufficient condition on the price update rates is that $k^{p}_{j}(\sum_{i \in\mathcal{I}_{j}}^{¥}k_{ij})$, where $\mathcal{I}_{j}=\{i\in \mathcal{I}: q_{ij}>0\}$, has a different value for each $j\in \mathcal{J}$. {This condition can be satisfied with probability 1, e.g. by drawing $k^{p}_{j}$'s and $k^{q}_{ij}$'s independently from some continuous distribution.}

\section{Numerical Results and Discussion}
\label{sec:discussion}

\begin{figure}[h]%
\centering

\begin{minipage}{3.5in}%
   \centering
   \includegraphics[width=3in]{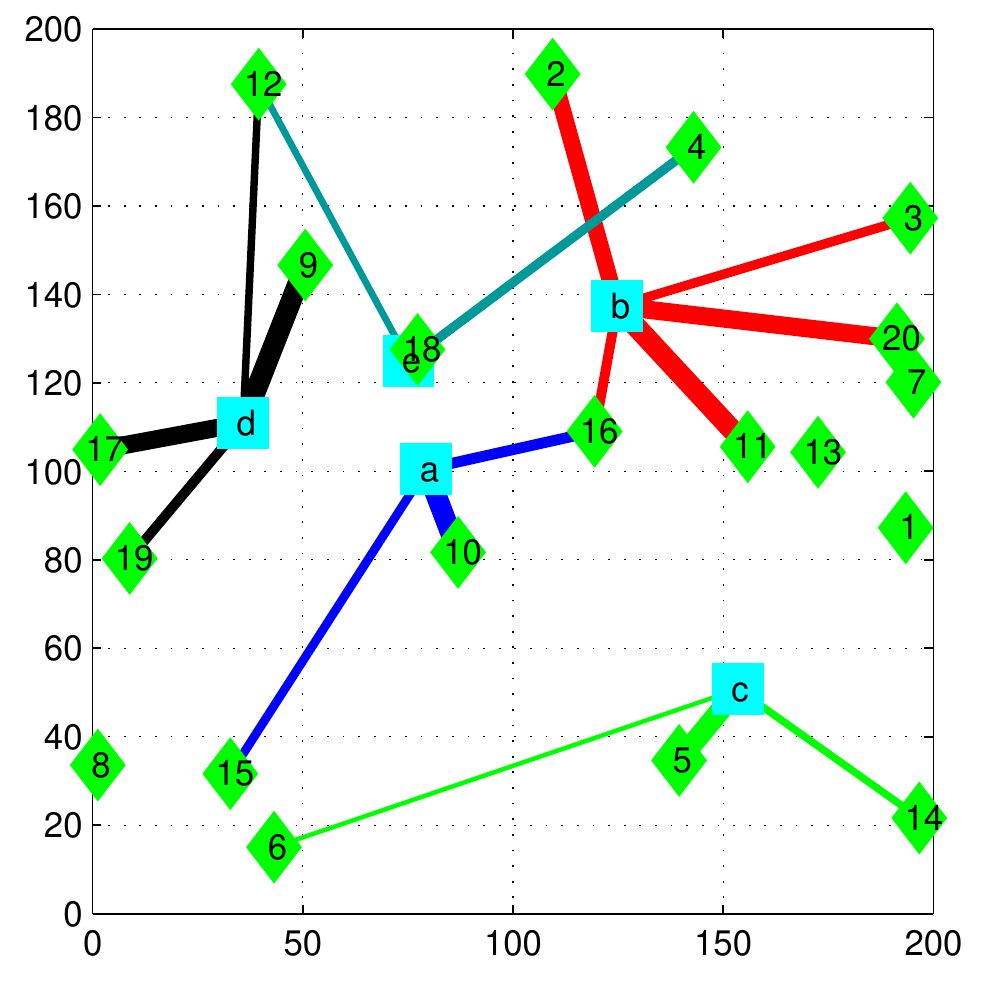} 
   \caption{Example of equilibrium user-provider association}
   \label{fig:equilibriumUserProviderAssociation}
\end{minipage}  
\end{figure}

\begin{figure}[h]%
\centering
\begin{minipage}{3.2in}%
   \centering
   \includegraphics[width=3.2in]{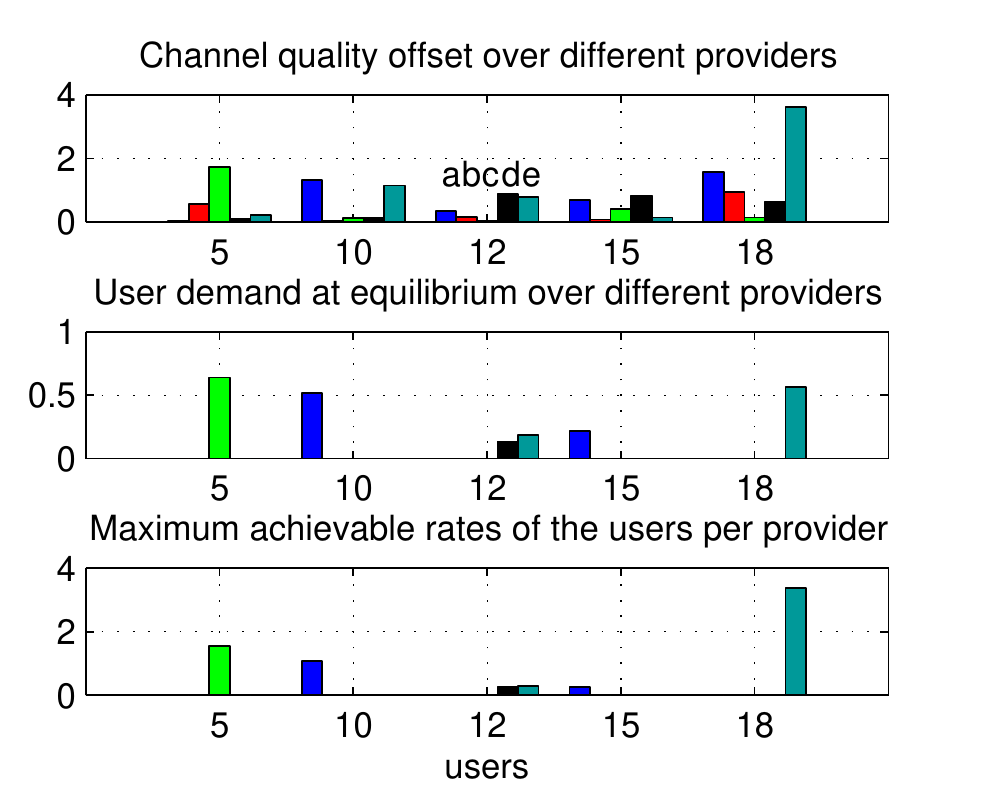} 
    \caption{Channel quality offset, demand and effective resource at equilibrium }
   \label{fig:userCapacitiesDemandRates}
\end{minipage}   
\qquad
\begin{minipage}{3.2in}%
   \centering
   \setlength{\fboxrule}{0.5pt} 
   \setlength{\fboxsep}{-0.2cm} 
{\includegraphics[width=3.2in]{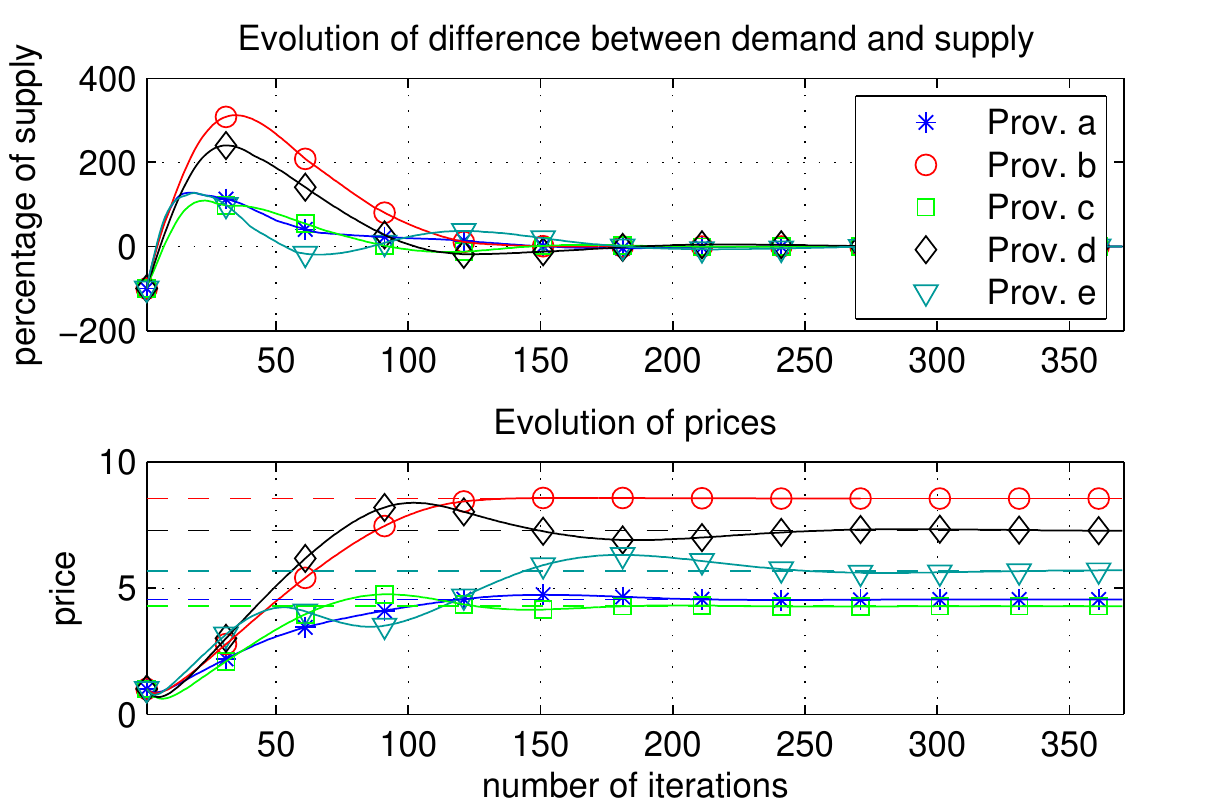}}
   \caption{Evolution of the primal-dual algorithm}
   \label{fig:evolutionPrimalDual}
\end{minipage}%
\end{figure}%

\begin{figure}
   \centering
\begin{minipage}{3.5in}%
   \centering
\includegraphics[width=3.2in]{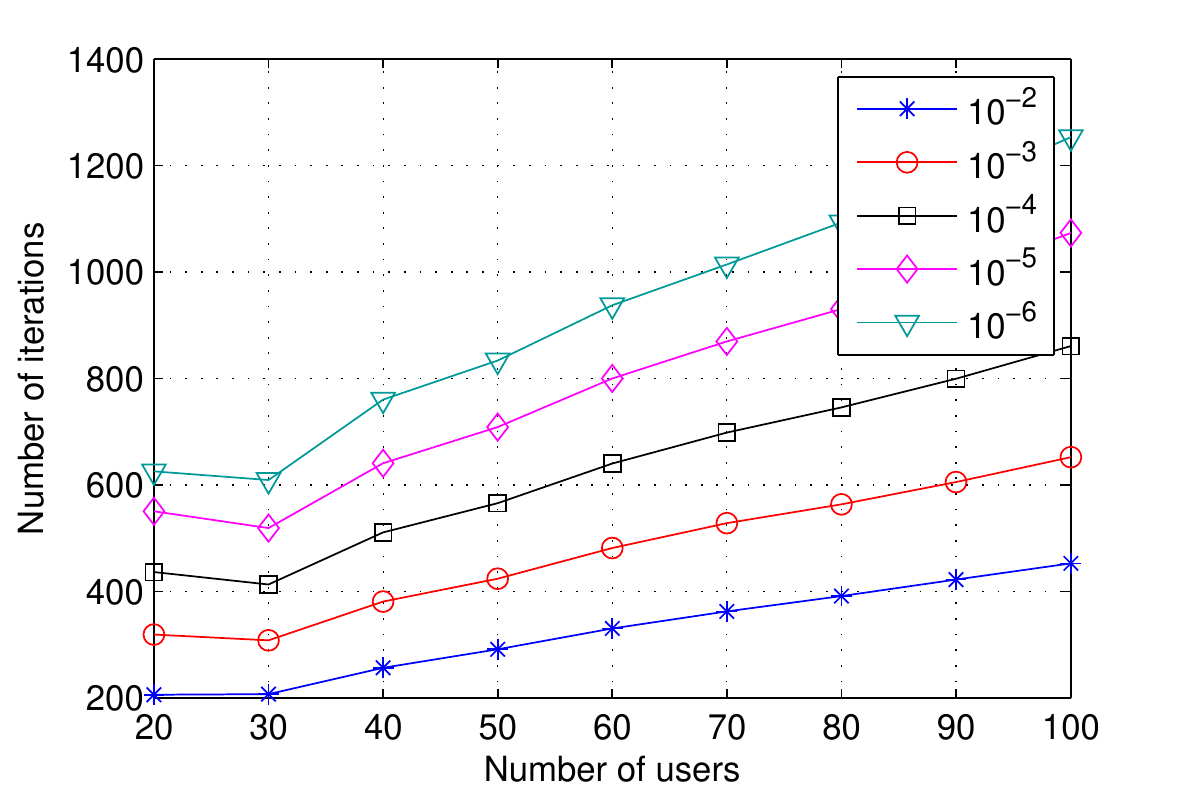}
   \caption{Average time to convergence, varying $\epsilon$}
   \label{fig:convergencePrecision}
\end{minipage}%
\end{figure}%

For numerical results, {we extend the setup from Example \ref{exa:TDMA}, where the resource being sold is the fraction of time allocated to exclusive use of the providers' frequency band, i.e., $Q_{j}=1$ for $j\in \mathcal{J}$.} We take the bandwidth of the providers to be $W_{j}=20MHz$, $j\in \mathcal{J}$. User $i$'s utility function is $a_{i}\log(1+\sum_{j=1}^{J}q_{ij}c_{ij})$, where we compute the spectral efficiency $c_{ij}$ from the Shannon formula $\frac{1}{2}W \log(1+\frac{E_{b}/N_{0} }{W}|h_{ij}|^{2})$, $q_{ij}$ is the allocated time fraction, $E_{b}/N_{0}$ is the ratio of transmit power to thermal noise, and $a_{i}$ is the individual willingness to pay factor taken to be the same across users. The channel gain amplitudes $|h_{ij}|=\frac{\xi_{ij}}{d^{\alpha/2}_{ij}}$ follow Rayleigh fading, where $\xi_{ij}$ is a Rayleigh distributed random variable with parameter $1$, and $\alpha=3$ is the outdoor power distance loss. 
{We choose the parameters so that} the $c_{ij}$ of a user is on average around $3.5 Mbps$ when the distance is $50m$, and around $60 Mbps$ when the distance is $5m$. The average signal-to-noise ratio $E_{b}/(N_{0}d^{\alpha})$  at $5m$ is around $25dB$.  We assume perfect modulation and coding choices such that the communication rates come from a continuum of values. The users are uniformly placed in a $200m$ by $200m$ area. {We want to emphasize that the above parameters are chosen for illustrative purposes only. Our theory applies to any number of providers, any number of users, any type of channel attenuation models, and arbitrary network topologies.}

\begin{figure}[h]%
\centering
\begin{minipage}{3.2in}%
   \centering
   \includegraphics[width=3.2in]{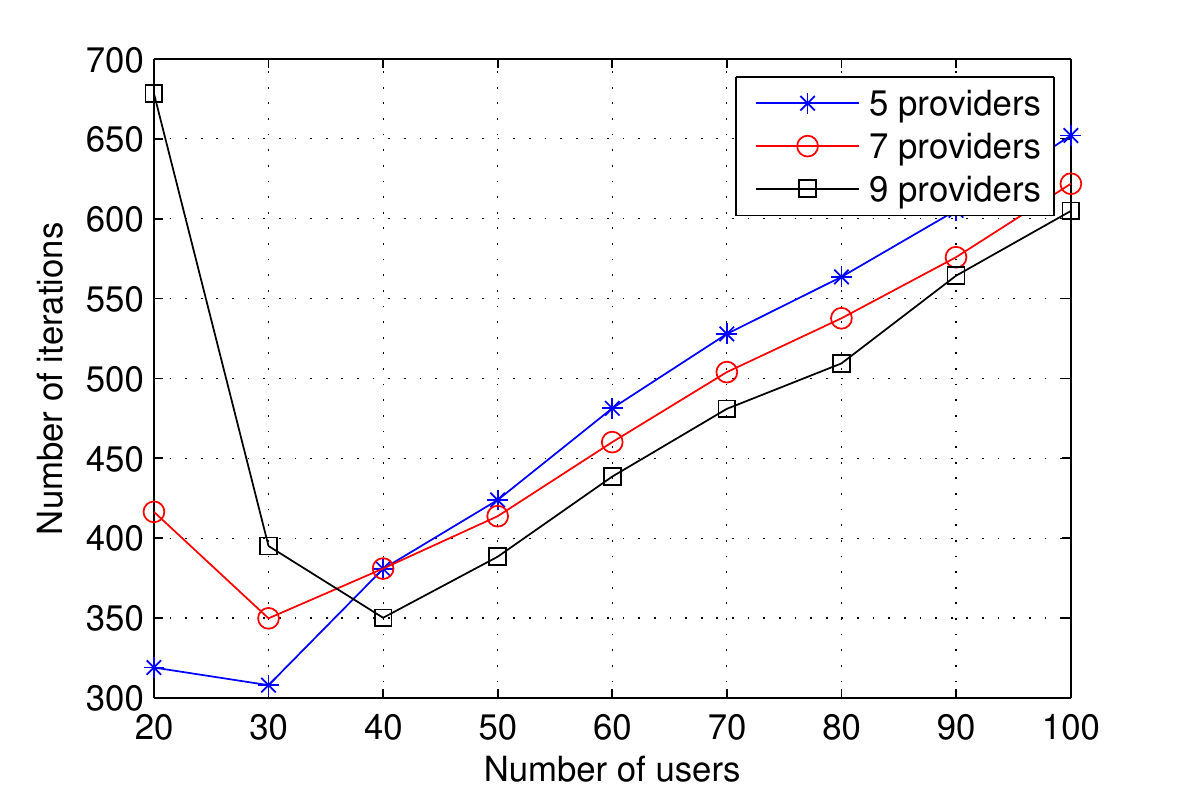} 
   \caption{Average time to convergence for different numbers of providers}
   \label{fig:convergenceDifferentProviderNumber}
\end{minipage}   
   \qquad \qquad
\begin{minipage}{3.5in}%
   \centering
   \setlength{\fboxrule}{0.5pt} 
   \setlength{\fboxsep}{-0.2cm} 
 \includegraphics[width=3.2in]{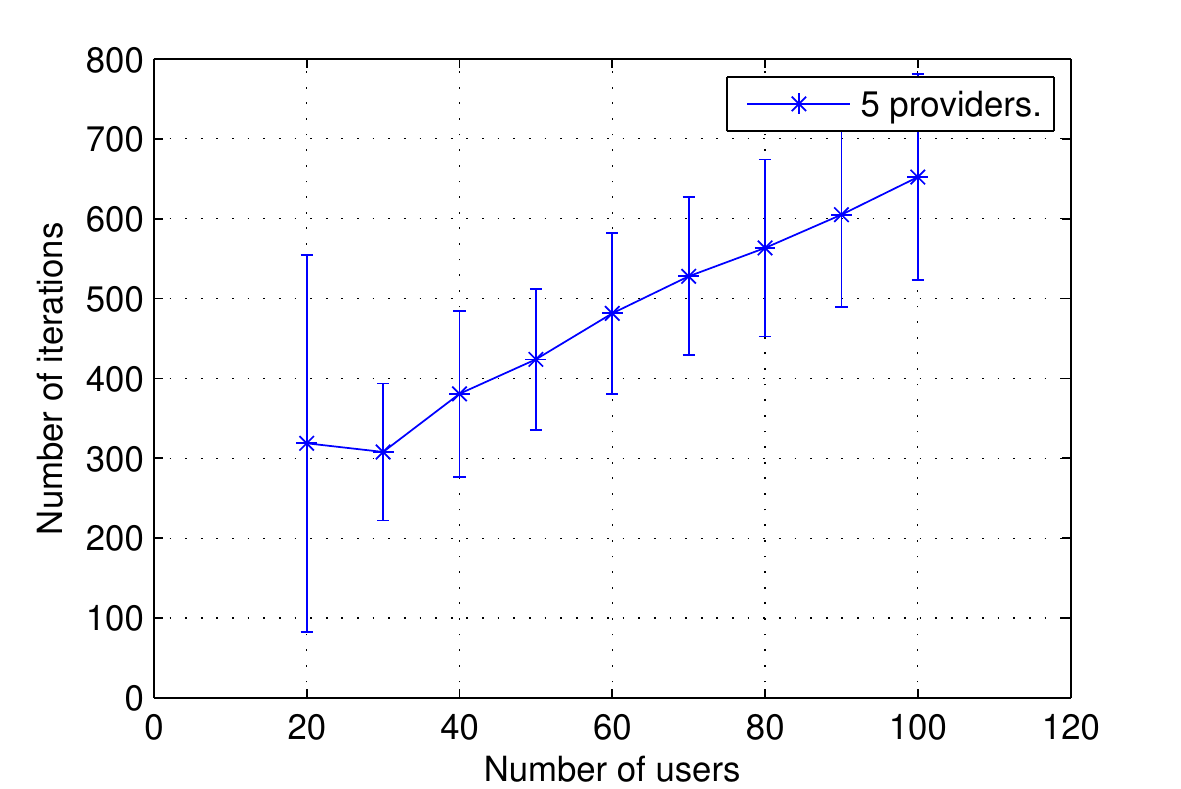}
   \caption{Average time to convergence with standard deviation}
   \label{fig:convergenceStandardDeviation}
\end{minipage}%
\end{figure}%

{We first consider a single instantiation with 20 users and 5 providers. In Fig. \ref{fig:equilibriumUserProviderAssociation}, we show the user-provider association at the equilibrium for a particular realization of channel gains, where the thickness of the link indicates the amount of resource purchased. The users are labeled by numbers (1-20), and the providers are labeled by letters ($a$-$e$). This figure shows two undecided users (12 and 16), and that certain users (1,7,13, and 8) do not purchase any resource at equilibrium. For the same realization of channel parameters, Fig. \ref{fig:userCapacitiesDemandRates} shows the channel quality, user demand, and rate at the equilibrium, for users 5, 10, 12, 15, and 18. We see that user 15 has better channel with provider $d$ than with provider $a$, but in the equilibrium all of his demand is towards provider $a$. This can be explained by looking at the bottom part of Fig. \ref{fig:evolutionPrimalDual}, where the dashed lines indicate equilibrium prices. We see that the provider $a$ announces a smaller price than provider $d$. The equilibrium prices reflect the competition among users: in Fig. \ref{fig:equilibriumUserProviderAssociation} we see that provider $b$ has the most customers, so it is not surprising that its price is the highest, as seen in Fig. \ref{fig:evolutionPrimalDual}.}

We next consider the convergence time of the {discrete time version of the} primal-dual algorithm. We fix the number of providers to be 5, and change the number of users from 20 to 100. For each parameter, we run 200 experiments with randomly generated user and provider locations and plot the average speed of convergence. The convergence is defined as the number of iterations after which the difference between supply and demand is no larger than $\epsilon Q_{j}$. Fig. \ref{fig:convergencePrecision} shows the average convergence time for different values of $\epsilon$. In general, 200 to 400 iterations are needed for convergence with $\epsilon=10^{-2}$, and 100-200 more iterations to get to $10^{-3}$. In Fig. \ref{fig:convergenceDifferentProviderNumber}, we compare the average convergence time for different number of providers. Here we take the stopping criterion to be $\epsilon=10^{-3}$. The convergence time depends on the update rates used for the primal-dual: if the update rates are too high, then the variables will tend to oscillate, so the algorithm will take too long to converge. On the other hand, if the rates are too small, the variables may converge too slowly. According to our theoretical analysis from Section \ref{sec:primal-dual}, we randomly assigned update variables to ensure global convergence of the algorithm. In general a very small or very large number of users per provider means that the algorithm will take longer to converge. Finally, Fig. \ref{fig:convergenceStandardDeviation} shows the average convergence time for 5 providers with the standard deviation. The variance of convergence time does not change with the  number of users, except for the case of 20 users. Empirically, the algorithm is sensitive to the choice of update rates when the ratio of users per provider is smaller than 4. In such cases the prices and demands may oscillate and take a long time to converge, which can also be seen in Fig. \ref{fig:convergenceDifferentProviderNumber}.

\subsection{Discussion}

We first comment on the implication of having undecided users at equilibrium. Given a set of prices, the decided users can calculate the unique demand vector that maximizes their payoffs, while the undecided users have an infinite number of such vectors. In particular, calculating the equilibrium maximizing demand vectors for undecided users may require cooperation between different providers, which may be challenging in practice.

On the other hand, the number of undecided users is small, i.e., not larger than $J$, and it does not grow with the number of users. Future systems may have user action replaced by the actions of software agents in charge of connection and handover between different providers. In this case, splitting over different providers may become feasible. This is similar to soft handoff (soft handover), a feature used by CDMA and WCDMA standards \cite{Wong:1997rw}. In addition, when the number of users is large, the impact of a single user on the price may be small. Hence, operating at a non-equilibrium price as the result of the decisions of a few undecided users may not have a great impact on the experienced quality of service, although the exact loss remains to be quantified.  

In the model we consider, we observe the ``locally monopolistic'' nature of wireless commerce, which does not exist for most other traditional goods. Namely, a user that has a strong channel to some provider, but a weak one to others, is willing to pay a higher price to the provider with the strong channel, and is thus not influenced by moderate price changes during the price competition. On the other hand, users with similar channel gains to all providers will be more sensitive to price competition. This local monopoly is in contrast to some other wireless resource allocation models where users' association is based solely on the price, in which case \emph{all} users go to the provider with the cheapest price (see, e.g. \cite{Maille:2008nh}).  

\section{Related Work}
\label{sec:related}
{In this paper we have considered a linear-usage pricing scheme, which has been widely adopted in the literature (see e.g. \cite{Kelly:1998, Shen:2004qa}).} Analyzing such pricing yields {various insights}: for example the existing TCP protocol can be interpreted as a usage-based pricing scheme that solves a network utility maximization problem \cite{Kelly:1998}. In practice, however, providers charge monthly subscription fees. For both voice and data plans, these subscriptions are sometimes combined with linear pricing beyond a predefined usage threshold. Pure linear pricing based on instantaneous channel conditions is generally not used, although it has received renewed attention due to near-saturation of some mobile networks (see, e.g. \cite{fierce2010:fk}). Recently, AT\&T introduced hybrid price plans, consisting of a flat rate fee for a certain amount of data, and linear pricing beyond that limit \cite{engadget2010:fk}.

There exists a rich body of related literature on using pricing and game theory to study provider resource allocation and interactions of service providers. The related research in the wireless setting can be divided into several categories: optimization-based resource allocation of one provider (e.g., \cite{Saraydar:2002,Marbach:2002fk,Chiang:2004,Acemoglu:2004bv,Musacchio:2006tg,Jiang:2008rm}), game theoretic study of interactions between the users of one provider (e.g., \cite{Adlakha:2007rw,Etkin:2005uq,Huang:2006,Kalathil:2010vn}), competition of different service providers on behalf of the users (e.g., \cite{Zhou:2005ys,Grokop:2008}), and providers' price competition to attract users (e.g., \cite{Zemlianov:2005, Sengupta:2007, Jia:2008, Ileri:2005a, Acharya:2009qd, Inaltekin:2007, Niyato:2009tg, Xi:2008nx, Gao:2008kx}). Our current work falls into the last category. 

{In our work, we have simultaneously considered several factors that reflect diverse wireless network scenarios: an arbitrary number of wireless providers compete for an arbitrary number of atomic users, where the users are heterogenous both in channels gains and in willingness to pay.} In related work where providers price-compete to attract users \cite{Zemlianov:2005, Jia:2008, Ileri:2005a}, purchasing a unit of resource from different providers brings the same amount of utility to a user; in our work a user's utility still depends on the channel gain to the provider. In other work where the Wardrop equilibrium concept is used (e.g., \cite{Maille:2008nh}), users are infinitesimal and non-heterogenous{; in our work users are atomic and have different willingness to pay}. One of the early work that explicitly takes into account the channel differences for different users on a line is \cite{Inaltekin:2007}, {for} infinitesimal users and distance-based channel gains. {Recently, a model similar to ours was used to treat a three-tier system \cite{Acharya:2009qd}, although for specific utility functions. The multiple-seller multiple-buyer dynamics in a cognitive radio setting was studied in \cite{Niyato:2009tg} using evolutionary game theory. Finally, \cite{Xi:2008nx} and \cite{Gao:2008kx} consider price competition in a multi-hop wireless network scenario. }

The design and proof of the {decentralized} algorithm were inspired by Chen et al. \cite{Chen:2008la}, with several key differences. First, their work considers the optimal resource allocation of a single OFDM cell. Second, it studies a system where each user has a total power constraint. Third, there {are often infinitely  many} global optimal solutions in \cite{Chen:2008la} and finally, our convergence results are proved with a set of conditions that are less stringent than those of \cite{Chen:2008la}.

\section{Conclusions and Future Work}
\label{sec:conclusion}

We provide an overview of the relationship between different concepts used throughout this work in Figure \ref{fig:concepts}. 

\colorlet{lightgreen}{green!20!yellow}
\colorlet{lightblue}{blue!20}
\colorlet{lightred}{red!20}
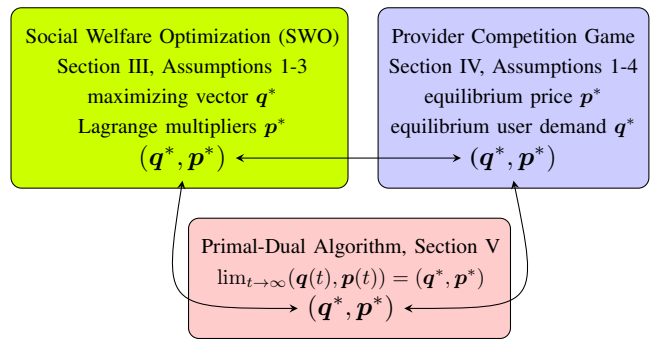
\begin{figure}[ht] 
   \centering
	\begin{tikzpicture}[>=stealth, scale=0.8]
		\draw[fill=lightgreen, rounded corners, transform shape] (-2.15,2) rectangle (3.45,5);
		\draw[fill=lightblue, rounded corners, transform shape] (3.95,2) rectangle (8.45,5);
		\draw[fill=lightred, rounded corners, transform shape] (0.8,-0.5) rectangle (6.15,1.5);		
		\draw (0.7,4.5) node [fill=lightgreen, rounded corners, transform shape] {Social Welfare Optimization (SWO)};
		\draw (0.7,4) node [fill=lightgreen, rounded corners, transform shape] {Section III, Assumptions 1-3};
		\draw (0.7,3.5) node (demand) [fill=lightgreen, rounded corners, transform shape] {maximizing vector $\boldsymbol{q}^{*}$};
		\draw (0.7,3) node (lagrange) [fill=lightgreen, rounded corners, transform shape] {Lagrange multipliers $\boldsymbol{p}^{*}$};	
		\draw (0.7,2.5) node (SWO) {$(\boldsymbol{q}^{*},\boldsymbol{p}^{*})$};	
		\draw (6.2,4.5) node [fill=lightblue, rounded corners, transform shape] {Provider Competition Game};
		\draw (6.2,4) node [fill=lightblue, rounded corners, transform shape] {Section IV, Assumptions 1-4};
		\draw (6.2,3.5) node (price) [fill=lightblue, rounded corners, transform shape] {equilibrium price $\boldsymbol{p}^{*}$}	;
		\draw (6.2,3) node (best response)[fill=lightblue, rounded corners, transform shape] {equilibrium user demand $\boldsymbol{q}^{*}$};
		\draw (6.2,2.5) node (PCG) {($\boldsymbol{q}^{*},\boldsymbol{p}^{*})$};			
		\draw (3.5,1) node [fill=lightred, rounded corners, transform shape] {Primal-Dual Algorithm, Section V};
		\draw (3.5,0.5) node [fill=lightred, transform shape] {$\lim_{t \rightarrow \infty}(\boldsymbol{q}(t),\boldsymbol{p}(t))=(\boldsymbol{q}^{*},\boldsymbol{p}^{*})$ };		
		\draw (3.5,0) node (PDA) {$(\boldsymbol{q}^{*},\boldsymbol{p}^{*})$};			
		\draw[<->] (SWO.east) -- (PCG.west);
		\draw[<->] (PDA.west) .. controls +(left:2.25cm) .. (SWO.south);
		\draw[<->] (PDA.east) .. controls +(right:2.25cm) .. (PCG.south);		
		\end{tikzpicture}
	\caption{Relationship between different concepts}
   	\label{fig:concepts}
\end{figure} 

We have studied the competition of an arbitrary number of wireless service providers, who want to serve a group of atomic users who are heterogenous in both willingness to pay and channel quality. We have modeled this interaction as a two-stage wireless provider game, and have characterized its unique equilibrium. We have shown that the provider competition leads to a unique socially optimal resource allocation for a broad class of utility functions and a generic channel model. Our results show that some users need to purchase their resource from several providers at the equilibrium, although the number of such users is upper bounded by the number of providers. We have also developed a decentralized algorithm which converges to the equilibrium prices as well as the equilibrium demand vectors using only local knowledge. 

Further work may include the study of fractional equilibria with the goal of characterizing the losses that occur when undecided users are unable to split their resource demand in an optimal way. It is also interesting to consider communication models where users cause externalities such as interference to each other.

\appendices
\renewcommand{\thesubsectiondis}{\arabic{subsection}.} 
\renewcommand{\thesubsection}{\Alph{section}-\arabic{subsection}}

\section{Social Optimality Proofs}
\subsection{Proof of BGR Property \ref{item:noLoops}}  
\label{app:diffSupport}

We first examine the properties of the optimal demand vector of the SWO problem. We will express the SWO in terms of the demand vector $\boldsymbol{q}$ only, by substituting directly equation (\ref{eqn:relative}) into equation (\ref{eqn:max}). Let $\boldsymbol{p}=[p_{1}\cdots p_{J}]$ be the vector of Lagrangian multipliers. The Lagrangian for SWO is then
\begin{align}
L(\boldsymbol{q},\boldsymbol{p})=& \sum_{i=1}^{I}u_{i}\left(\sum_{j=1}^{J}q_{ij} c_{ij}\right) + \sum_{j=1}^{J}p_{j}\left( Q_{j}-\sum_{i=1}^{I}q_{ij}\right).
\end{align}
It is easy to check that the SWO problem satisfies the Slater's condition \cite{Boyd:2004}, and thus the sufficient and necessary KKT conditions for an optimal solution $(\boldsymbol{q},\boldsymbol{p})$ are as follows:

\begin{align}
\frac{\partial u_{i}(x_{i})}{\partial x_{i}}c_{ij}-p_{j} \leq &0,\; j \in \mathcal{J};\; i \in \mathcal{I} \label{eqn:kkt1}\\
q_{ij}\left(\frac{\partial u_{i}(x_{i})}{\partial x_{i}}c_{ij}-p_{j} \right)=&0,\; j \in \mathcal{J};\; i \in \mathcal{I} \label{eqn:kkt2}\\
\sum_{j=1}^{J}q_{ij} c_{ij} = & x_{i},\;  i \in \mathcal{I} \label{eqn:kkt3}\\
\sum_{i=1}^{I}q_{ij}=&Q_{j},\; j \in \mathcal{J} \label{eqn:kkt4}\\
p_{j}>0,\; q_{ij} \geq& 0\;j \in \mathcal{J};\; i \in \mathcal{I} \label{eqn:kkt5}
\end{align}
where with some abuse of notation we use $\frac{\partial u_{i}(x_{i})}{\partial x_{i}}$ to denote $\frac{\partial u_{i}(x)}{\partial x}\Big\vert_{x=x_{i}=\sum_{j=1}^{J}q_{ij} c_{ij}}$. 

The following characterizes the relationship between the prices of any two networks from which a user has strictly positive demand.

Recall the support set definition $\hat{\mathcal{J}}_{i}(\boldsymbol{q}_{i})=\{j \in \mathcal{J}: q_{ij}>0\}$. From \eqref{eqn:kkt1} we see that $\frac{\partial u_{i}(x_{i})}{\partial x_{i}} \leq \min_{k\in \mathcal{J}} \frac{p_{k}}{c_{ik}}$. Then, from  (\ref{eqn:kkt2}) we can see that $q_{ij}>0$ only when $\frac{\partial u_{i}(x_{i})}{\partial x_{i}} = \frac{p_{j}}{c_{ij}}$. Hence $\frac{p_{j}}{c_{ij}} = \min_{k \in \mathcal{J}}\frac{p_{k}}{c_{ik}} $ is a necessary condition for $q_{ij}>0$ for all $i \in \mathcal{I}, j\in \mathcal{J}$.
Then, $q_{ij}>0$ and $q_{ij'}>0$ implies $\frac{p_{j}}{c_{ij}}=\frac{p_{j'}}{c_{ij'}}=\min_{k \in \mathcal{J}}\frac{p_{k}}{c_{ik}}$. In particular, $q_{ij}>0$ and $q_{ij'}>0$ implies
\begin{equation}
\label{eqn:indifferent}
\frac{c_{ij}}{c_{ij'}}=\frac{p_{j}}{p_{j'}}.
\end{equation}

We now consider the BGR defined by the support sets $\{\hat{\mathcal{J}}_{i}\}_{i=1}^{I}$ of undecided users. For any two edges $(i,j)$ and $(i,k)$ of BGR, where $i$ is the user index and $j,k$ are the provider indices, $q_{ij}>0$ and $q_{ik}>0$ so by (\ref{eqn:indifferent}) we have $\frac{c_{ij}}{c_{ik}}=\frac{p_{j}}{p_{k}}$.

Suppose that a loop exists in BGR (refer to Fig. \ref{fig:BGRloop} for this part of the proof). Then, a sequence of nodes $i_{1},j_{1},i_{2},j_{2},\ldots,$ $i_{n},j_{n},i_{1}$ exists, where $i_{1},\ldots,i_{n}$ are the user nodes and $j_{1},\ldots,j_{n}$ are the provider nodes, such that $(i_{k},j_{k})$ and $(j_{k-1},i_{k})$ are edges in BGR for $k=1,\ldots n$ (with $i_{0}$ defined as $i_n$). We assume that the members of the sequence are distinct otherwise there is already a smaller loop inside. {Since both $(i_{k},j_{k})$ and $(j_{k-1},i_{k})$ are edges, then $\frac{c_{ik-1}}{c_{ik}}=\frac{p_{k-1}}{p_{k}}$, based on (\ref{eqn:indifferent}). A loop in the BGR implies:}

{\begin{align*} 
\frac{c_{i_{1}n}}{c_{i_{1}1}}\frac{c_{i_{2}1}}{c_{i_{2}2}} \ldots \frac{c_{i_{n-1}n-2}}{c_{i_{n-1}n-1}}\frac{c_{i_{n}n-1}}{c_{i_{n}n}}=\frac{p_{n}}{p_{1}} \frac{p_{1}}{p_{2}} \ldots \frac{p_{n-2}}{p_{n-1}} \frac{p_{n-1}}{p_{n}}=1.
\end{align*}}

\tikzstyle{user}=[circle,draw=black,thin, minimum size=2pt, transform shape] 
\tikzstyle{provider}=[rectangle,draw=black,fill=white, minimum size=5mm, transform shape] 
\begin{figure}[ht] 
	   \centering
	\begin{tikzpicture}[>=stealth, scale=0.8]
		\node (user1) at (0,2) [user] {$i_{1}$};
		\node (user2) at (1.5,2) [user] {$i_{2}$};
		\node (user3) at (3,2) [user] {$i_{3}$};
		\node (userN) at (6,2) [user] {$i_{n}$};
		\node (net1) at (0,0) [provider] {$j_{1}$};
		\node (net2) at (1.5,0) [provider] {$j_{2}$};
		\node (net3) at (3,0) [provider] {$j_{3}$};
		\node (net7) at (5.5,0) [provider] {$j_{n-1}$};
		\node (netN) at (7,0) [provider] {$j_{n}$};
		\draw [-] (user1) -- (net1);
		\draw [-] (user1) -- (netN);
		\draw [-] (user2) -- (net1);
		\draw [-] (user2) -- (net2);
		\draw [-] (user3) -- (net2);
		\draw [-] (user3) -- (net3);	
		\draw  (4.25,1) node {$\cdots$};
		\draw [-] (userN) -- (net7);
		\draw [-] (userN) -- (netN);
		\draw [-] (user3) -- (net2);
		\draw [-] (user3) -- (net3);			
	\end{tikzpicture}
	\caption{A bipartite graph representation loop}
   	\label{fig:BGRloop}
	\end{figure}
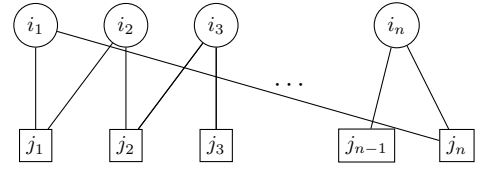
{Since $\frac{c_{i_{1}n}}{c_{i_{1}1}}\frac{c_{i_{2}1}}{c_{i_{2}2}} \ldots \frac{c_{i_{n-1}n-2}}{c_{i_{n-1}n-1}}\frac{c_{i_{n}n-1}}{c_{i_{n}n}}$ is a function of independent continuous random variables, it is also a continuous random variable itself. The probability that the product of independent continuous random variables equals a constant is zero, so we can conclude that a BGR has loops with probability zero. In other words, a BGR has no loop with probability one.}

\subsection{BGR Algorithm}
\label{app:algo}

Let $E$ be the set of edges, and $\hat{\mathcal{I}}$ and $\hat{\mathcal{J}}$ be the set of all user and provider nodes, respectively, present in the BGR. The demand of undecided users can be found using Algorithm 1. 
\begin{algorithm}[h!] 
\begin{algorithmic}[1] 
\STATE For each undecided node $i \in \hat{\mathcal{I}}$, calculate the checksum $P_{i} \leftarrow x_{i}^{*}$
\STATE For each provider $j \in \hat{\mathcal{J}}$ calculate the checksum $S_{j} \leftarrow Q_{j}-\sum_{i:(i,j) \notin \mathcal{G}}^{¥}q_{ij}^{*}$, $\forall j \in \hat{\mathcal{J}}$
\STATE For each $q_{ij}^{*}>0$, add edge (i,j) to the edge set E
\medskip 
\WHILE{E $\neq \emptyset$ } 
\STATE \textbf{find} a leaf node $l$ and associated edge $(i,j)$ \label{alg:leaf}
\IF {\text{the leaf node is a user node}}
\STATE $q_{ij}^{*} \leftarrow \frac{P_{i}}{c_{ij}}$ 
\ELSE 
\STATE $q_{ij}^{*} \leftarrow S_{i}$ 
\ENDIF 
\STATE $P_{i} \leftarrow \left(P_{i}-q_{ij}^{*}c _{ij}\right)$ and $S_{j} \leftarrow \left(S_{j}-q_{ij}^{*}\right)$
\STATE remove edge $(i,j)$
\ENDWHILE
\end{algorithmic} 
\caption{BGR decoding}\label{alg:BGR} 
\end{algorithm} 
 
We now give an informal description of an algorithm that finds the {optimal and unique} values of $\boldsymbol{q}_{i}^{*}$ for undecided users. Since BGR has no loops, it is a (unrooted) tree. Hence, we can run a simple iterative algorithm which removes a leaf node (node with a single incoming edge) and its associated edge at each iteration. We begin by finding a leaf node in the BGR. We then determine the demand of the edge associated to the leaf node either from BGR Property \ref{item:checkUser}) or \ref{item:checkNetwork}). Using this value we update the check-sum value of its parent node. Then we remove the leaf node and the associated edge. This completes one iteration. We repeat the process until there are no more edges in the graph. 

The key for Algorithm \ref{alg:BGR} to work is that the BGR has no loops, so a leaf node can always be found in line \ref{alg:leaf}. Notice that in the last iteration, there will be only one user node $i$ and one provider node $j$ left connected by an edge with value $q_{ij}^{*}$. The checksums for these two nodes are $P_{i}$ and $S_{j}$, which satisfy $P_{i}=S_{j}c_{ij}$ since $P_{i}=q_{ij}^{*}c _{ij}$ and $S_{j}=q_{ij}^{*}$. Upon completion of the algorithm, the demand of undecided users is uniquely defined.

\subsection{Proof of Theorem~\ref{thm:uniqueQ}}
\label{app:proofUniqueDemand}
Assume there exist two optimal demand vectors of SWO $\boldsymbol{q}^{*}$ and $\boldsymbol{q}'$. By Lemma \ref{lem:diffSupport}, $\boldsymbol{q}^{*}$ and $\boldsymbol{q}'$ have different support sets $\{\hat{\mathcal{J}}_{i}^{*}\}_{i=1}^{I}$ and $\{\hat{\mathcal{J}}_{i}'\}_{i=1}^{I}$ almost surely. Next, consider a convex combination demand vector $\boldsymbol{q}^{\lambda}=\lambda \boldsymbol{q}^{*}+\bar{\lambda}\boldsymbol{q}'$ where $\lambda \in (0,1)$, $\bar{\lambda}=1-\lambda$. Since $x_{i}^{*}=\sum_{j=1}^{J}q_{ij}^{*}c _{ij}=\sum_{j=1}^{J}q_{ij}'c _{ij}$, then $\sum_{j=1}^{J}q_{ij}^{\lambda}c_{ij}=\lambda\sum_{j=1}^{J}q_{ij}^{*}c _{ij}+\bar{\lambda}\sum_{j=1}^{J}q_{ij}'c _{ij}=x_{i}^{*}$, so it follows that $\boldsymbol{q}^{\lambda}$ is also a maximizing solution of SWO for any $\lambda \in (0,1)$. Then, the support set $\hat{\mathcal{J}}_{i}^{\lambda}(\boldsymbol{q}^{\lambda})=\{j\in\mathcal{J}: q^{\lambda}_{ij}=\lambda q_{ij}^{*} + \bar{\lambda}q_{ij}'>0\}$ for user $i$ is $\hat{\mathcal{J}}_{i}^{\lambda}=\hat{\mathcal{J}}_{i}^{*} \cup \hat{\mathcal{J}}_{i}' $, for all $\lambda \in (0,1)$. In particular, the support sets $\{\hat{\mathcal{J}}_{i}^{\lambda}\}_{i=1}^{I}$ are the same for all $\lambda \in (0,1)$, which is a contradiction to Lemma \ref{lem:diffSupport}. 

\section{Provider Competition Game Proofs}
\subsection{Proof of Lemma \ref{lem:uniqueEffectiveUPM}}
\label{app:proofUPM}

{It can be verified that UPM satisfies Slater's conditions \cite{Boyd:2004}. The necessary and sufficient KKT conditions for an optimal solution $\boldsymbol{q}_{i}\geq \boldsymbol{0}$ of UPM of user $i$ are as follows:
\begin{align}
u'_{i}(x_{i})c_{ij}\leq &p_{j},\; j \in \mathcal{J}  \label{eqn:UPMcondition1}\\ 
q_{ij}\left(u'_{i}(x_{i})c_{ij}- p_{j}\right)=&0,\; j \in \mathcal{J} \label{eqn:UPMcondition2} \\
\text{where }x_{i}=\sum_{j=1}^{J}q_{ij}c_{ij},\; \boldsymbol{q}_{i}\geq \boldsymbol{0} \label{eqn:UPMcondition3}
\end{align}
Expression (\ref{eqn:UPMcondition1}) implies $u_{i}'(x_{i})\leq \alpha$, where $\alpha=\min_{k\in \mathcal{J}}\frac{p_{k}}{c_{ik}}$. Based on the utility function of user $i$, there are two cases: $u_{i}'(0) < \alpha$ and $u_{i}'(0) \geq \alpha$. 

In the first case $u_{i}'(0)c_{ij}-p_{j}<0$, so $u_{i}'(x_{i})c_{ij}-p_{j}<0$ for all $x_{i}\geq 0$ and for all $j \in \mathcal{J}$. This is because the marginal utility $u'_{i}(\cdot)$ is a strictly decreasing function, by Assumption \ref{as:concaveUtility}. Thus, by \eqref{eqn:UPMcondition2}, $q_{ij}=0$ for all $j \in \mathcal{J}$. Therefore, $\boldsymbol{q_{i}}=\boldsymbol{0}$ and by \eqref{eqn:UPMcondition3} $x_{i}^{*}=0$. So, \eqref{eqn:UPMcondition1}-\eqref{eqn:UPMcondition3} hold for a unique value $x_{i}^{*}=0$.

In the second case, $u_{i}'(0) \geq \alpha$. Then, because $u'_{i}(\cdot)$ decreases to zero (Assumption \ref{as:concaveUtility}), there is a unique $\hat{x}_{i}\geq 0$ such that $u'_{i}(\hat{x}_{i})=\alpha$. We first check that there is a $\boldsymbol{q}_{i}$ such that equations \eqref{eqn:UPMcondition1}-\eqref{eqn:UPMcondition3} hold with $x_{i}=\hat{x}_{i}$. Equation \eqref{eqn:UPMcondition1} holds because $u_{i}'(\hat{x}_{i})=\alpha\leq{p_{j}}/{c_{ij}}$  for all $j \in \mathcal{J}$. Next, by \eqref{eqn:UPMcondition2}, for any $j$ such that ${p_{j}}/{c_{ij}}>\alpha=u_{i}'(\hat{x}_{i})$ we have $q_{ij}=0$. For any other $j$, ${p_{j}}/{c_{ij}}=\alpha=u_{i}'(\hat{x}_{i})$, thus, with respect to \eqref{eqn:UPMcondition2}, $q_{ij}$ can take any non-negative value. In particular, for the set $\{j\in \mathcal{J}:p_{j}/c_{ij}=\alpha\}$ we can choose $q_{ij}$'s so that \eqref{eqn:UPMcondition3} holds.

Note that $q_{ij}$ is positive only when $p_{j}/c_{ij}=\alpha$, which proves the last part of the lemma. It remains to show that $\hat{x}_{i}$ is the only value of $x_{i}$ for which a $\boldsymbol{q}_{i}$ satisfying \eqref{eqn:UPMcondition1}-\eqref{eqn:UPMcondition3} exists. For any $x_{i} < \hat{x}_{i}$, $u_{i}'(x_{i})> \alpha$ so \eqref{eqn:UPMcondition1} is violated for $j \in \argmin p_{k}/c_{ik}$. For any $x_{i}>\hat{x}_{i}$, $u_{i}(x_{i})<\alpha$ so $u_{i}(x_{i})c_{ij}-p_{j}<0$ for all $j \in \mathcal{J}$. Then, \eqref{eqn:UPMcondition2} implies that $q_{ij}=0$ for all $j \in \mathcal{J}$, meaning that $x_{i}=0$, which contradicts  $x_{i}>\hat{x}_{i}>0$. Therefore $\hat{x}_{i}$ is the unique searched value $x_{i}^{*}$. 
}

\subsection{Proof of Theorem \ref{thm:uniqueNash}}
\label{app:proofSPE}
Assume that the providers charge prices $\boldsymbol{p}=[p_{1}\ldots p_{J}]$ to the users. Then, each user faces a local maximization problem UPM$_{i}(\boldsymbol{p})$, as defined in (\ref{eqn:defUPM}).

{Equations (\ref{eqn:UPMcondition1}) - (\ref{eqn:UPMcondition3}), together with $\boldsymbol{q}_{i}\geq\boldsymbol{0}$ for all $ i\in \mathcal{I}$, are equivalent to equations (\ref{eqn:kkt1})-(\ref{eqn:kkt3}) and (\ref{eqn:kkt5}). Furthermore, under assumption \ref{as:RRA}, for $\boldsymbol{q}$ to be an SPE of the provider competition game, the demand to each provider must equal its supply, i.e., $\sum_{i=1}^{I}q_{ij}=Q_{j}$ for all $j \in \mathcal{J}$ (i.e., equations \eqref{eqn:kkt4}). Hence, the SPE is a price vector tuple $(\boldsymbol{p},\boldsymbol{q})$ that satisfies KKT conditions (\ref{eqn:kkt1})-(\ref{eqn:kkt5}). But, the KKT conditions (\ref{eqn:kkt1})-(\ref{eqn:kkt5}) are necessary and sufficient for any vector tuple $(\boldsymbol{p},\boldsymbol{q})$ to be the maximizing solution of SWO. Hence, we have established formal equivalence between the SPE of the provider competition game and the maximizing demand vector and Lagrangian multipliers of the SWO problem $(\boldsymbol{q}^{*},\boldsymbol{p}^{*}$). Hence, $(\boldsymbol{p}^{*}, \boldsymbol{q}^{*})$ forms the unique SPE of the provider competition game.}

\subsection{Proof of Lemma \ref{lem:nbrUndecided}} 
\label{app:proofNumberUndecided}
Given an arbitrary price vector $\boldsymbol{p}$, we can construct the corresponding BGR$(\boldsymbol{p})$ from the users' preference sets. First consider a BGR that is connected (single component). We start drawing the graph with a single undecided user node, and then add the provider nodes that are connected to this user node. There should be at least two such provider nodes. Then we add another undecided user node which shares one common provider node with the existing undecided user node. This new undecided user node will bring at least one new provider node into the graph, otherwise it leads to a loop in the graph. We repeat this process iteratively. Since the number of provider nodes is upper-bounded by $J$, the total number of undecided user nodes is upper-bounded by $J-1$. The ``$-1$'' is due to the fact that the first undecided user node is connected to (at least) two new provider nodes in the graph. If we consider a BGR with $b$ disconnected subgraphs, it can be shown that the total number of undecided users is bounded by $J-b$. 

\section{Primal Dual algorithm supporting proofs}

\subsection{Proof of Lemma~\ref{lem:decreasingLaSalleFunction} --- Proof that $\dot{V}\leq 0$: }

\label{app:proofVdotLessThanZero}

For the optimal demand vector $\boldsymbol{q}^{*}$ of the SWO problem and the associated Lagrange multipliers $\boldsymbol{p}^{*}$, we see that $f_{ij}^{*}=p_{j}^{*}$ whenever $q_{ij}^{*}>0$ from equation (\ref{eqn:kkt2}), and $f_{ij}^{*}\leq p_{j}^{*}$ when $q_{ij}^{*}=0$ from equations (\ref{eqn:kkt1}) and (\ref{eqn:kkt2}). Similarly, equation (\ref{eqn:kkt2}) ensures that $f_{ij}^{*}<p^{*}_{j}$ implies $q_{ij}^{*}=0$. This fact will be used shortly.  
For our La Salle function:
\begin{align*}
&\dot{V}=	\\ \stackrel{(a)}{=}
&\sum_{i,j}^{¥}(q_{ij}-q_{ij}^{*})(f_{ij}-p_{j})_{q_{ij}}^{+}+\sum_{j}^{¥}(p_{j}-p_{j}^{*})(\sum_{i}^{¥}q_{ij}-Q_{j})_{p_{j}}^{+} \nonumber \\
		\stackrel{(b)}{\leq}&\sum_{i,j}^{¥}(q_{ij}-q_{ij}^{*})(f_{ij}-p_{j})+\sum_{j}^{¥}(p_{j}-p_{j}^{*})(\sum_{i}^{¥}q_{ij}-Q_{j}) \nonumber\\
		\stackrel{(c)}{=}& \sum_{i}^{¥}(\sum_{j}(q_{ij}-q_{ij}^{*})(f_{ij}-f_{ij}^{*}))+\sum_{i,j}^{¥}(q_{ij}-q_{ij}^{*})(f_{ij}^{*}-p_{j}^{*}), 
\end{align*}
{where $(a)$ follows from the definition of $V$ and $\dot{V}$, $(b)$ can be readily verified by examining all the cases, and $(c)$ is obtained by some algebraic manipulation. The expression for $\dot{V}$ is now in such a form that we can prove $\dot{V}\leq 0$. 
}

First, the following is true for any two vectors $\boldsymbol{q}_{1}$ and $\boldsymbol{q}_{2}$ due to concavity (see, e.g. Section 3.1.3 of \cite{Boyd:2004}):
\begin{align}
& \nabla u_{i}^{T}(\boldsymbol{q}_{1})(\boldsymbol{q}_{2}-\boldsymbol{q}_{1})\geq u_{i}(\boldsymbol{q}_{2})- u_{i}(\boldsymbol{q}_{1}), \label{eqn:conc1} \\
& \nabla u_{i}^{T}(\boldsymbol{q}_{2})(\boldsymbol{q}_{1}-\boldsymbol{q}_{2})\geq u_{i}(\boldsymbol{q}_{1})- u_{i}(\boldsymbol{q}_{2}). \label{eqn:conc2}
\end{align}

Substituting $\boldsymbol{q}_{1}=\boldsymbol{q}_{i}(t)$ and $\boldsymbol{q}_{2}=\boldsymbol{q}_{i}^{*}$ in \eqref{eqn:conc1} and \eqref{eqn:conc2} gives:
\begin{align*}
\nabla u_{i}^{T}(\boldsymbol{q}_{i})(\boldsymbol{q}_{i}^{*}-\boldsymbol{q}_{i})\geq\nabla u_{i}^{*}(\boldsymbol{q}_{i})(\boldsymbol{q}_{i}^{*}-\boldsymbol{q}_{i}),
\end{align*}
which can be rewritten as 
\begin{align}
\sum_{j}(q_{ij}(t)-q_{ij}^{*})(f_{ij}(t)-f_{ij}^{*})\leq 0, \text{ for all } i\in \mathcal{I}, \label{eqn:invariant1}
\end{align}
which we recognize as one of the components of the first term in (\ref{eqn:VdotLessThenZero}). Considering a component of the second term in (\ref{eqn:VdotLessThenZero}), $(q_{ij}(t)-q_{ij}^{*})(f_{ij}^{*}-p_{j}^{*})$, we first recall from the KKT conditions \eqref{eqn:kkt2} that either $f_{ij}^{*}=p_{j}^{*}$, in which case $(q_{ij}-q_{ij}^{*})(f_{ij}^{*}-p_{j}^{*})=0$, or $f_{ij}^{*}<p_{j}^{*}$, in which case $q_{ij}^{*}=0$ so $(q_{ij}-q_{ij}^{*})(f_{ij}^{*}-p_{j}^{*})\leq 0$. Hence, 
\begin{align}
(q_{ij}(t)-q_{ij}^{*})(f_{ij}^{*}-p_{j}^{*})\leq 0, \text{ for all } i \in \mathcal{I},\; j\in \mathcal{J}. \label{eqn:invariant2}
\end{align}
which completes the proof that $\dot{V}\leq 0 $. 

\subsection{Proof of Lemma~\ref{lem:convergesToEquilMarginUtil} --- Proof that marginal utility on $V_{L}$ is the equilibrium marginal utility.}
\label{app:MarginalUtilEqualEquilUtil}
Recall that for user $i$, $f_{ij}(\boldsymbol{q}_{i})=\frac{\partial u_{i}(x_{i})}{\partial x_{i}}c_{ij}=u_{i}'({x}_{i})c_{ij}$ for all $j \in \mathcal{J}$, where ${x}_{i}(t)=\sum_{j}^{¥}q_{ij}(t)c_{ij}$ and $u_{i}'(x_{i}(t))$ is a \textbf{scalar} function of time. Hence, 
\[\nabla u_{i}^{T}(\boldsymbol{q}_{i})=u_{i}'(x_{i})[c_{i1} \cdots c_{iJ}]=u_{i}'(x_{i}) \boldsymbol{c}_{i}^{T}.\]

We now consider equations (\ref{eqn:sale1}), which can be rewritten as: 
\begin{align*}
\nabla u_{i}^{T}(\boldsymbol{q}_{i})(\boldsymbol{q}_{i}-\boldsymbol{q}_{i}^{*})=&\nabla u_{i}^{T}(\boldsymbol{q}_{i}^{*})(\boldsymbol{q}_{i}-\boldsymbol{q}_{i}^{*}), \text{ or,} \\
u_{i}'(x_{i}) \boldsymbol{c}_{i}^{T}(\boldsymbol{q}_{i}-\boldsymbol{q}_{i}^{*})=&u_{i}'(x_{i}^{*}) \boldsymbol{c}_{i}^{T}(\boldsymbol{q}_{i}-\boldsymbol{q}_{i}^{*}),
\end{align*}
which leads to $u_{i}'(x_{i}) \Delta{x}_{i}(t) =u_{i}'(x_{i}^{*}) \Delta{x}_{i} (t)$, where $\Delta{x}_{i}(t)=\boldsymbol{c}_{i}^{T}(\boldsymbol{q}_{i}(t)-\boldsymbol{q}_{i}^{*})$. Then, either $\Delta x_{i}(t)=0$, or $u_{i}'(x_{i}(t))=u_{i}'(x_{i}^{*})$. The necessary and sufficient condition for both is that that $x_{i}(t)=x_{i}^{*}$, for $\boldsymbol{q}_{i}(t)$, $i \in \mathcal{I}$ on the invariant set $V_{L}$. An immediate corollary is that $f_{ij}(t)=f_{ij}^{*}$ on $V_{L}$ for all $i\in \mathcal{I},j\in \mathcal{J}$.

The second part of the Lemma is simpler to prove. Equation (\ref{eqn:sale2}) gives $q_{ij}(t)(f_{ij}^{*}-p_{j}^{*})=q_{ij}^{*}(f_{ij}^{*}-p_{j}^{*})$, which is equal to zero by equation (\ref{eqn:kkt2}). This concludes the proof. 

\subsection{Proof of Lemma~\ref{lem:borderq} --- Proof that the support set of primal and dual variables does not change on the invariant set}
\label{app:supportConstantOnInvariant}
We prove that the set of positive demands does not change over time on the invariant set by contradiction. Suppose that $q_{ij}(t)>0$ but $q_{ij}(t+\tau)=0$ for all $\tau$ such that $0<\tau<\epsilon$, where $\epsilon$ is a small number. Then $\dot{q}_{ij}(t)<0$ but $\dot{q}_{ij}(t+\tau)=0$, so we have 
\begin{align*}
\lim_{\tau \rightarrow 0} \sum_{j \in  \mathcal{J}_{i}}^{¥}q_{ij}(t)c_{ij}-q_{ij}(t+\tau)c_{ij}<0
\end{align*}
On the other hand differentiating equation (\ref{eqn:sale3}) with respect to time yields $\sum_{j \in  \mathcal{J}_{i}}^{¥}q_{ij}(t)c_{ij}=0$ and $\sum_{j \in  \mathcal{J}_{i}}^{¥}q_{ij}(t+\tau)c_{ij}=0$, which is a contradiction. So a non-zero $q_{ij}$ stays non-zero. Suppose now that $q_{ij}(t)=0$ but $q_{ij}(t+\tau_{ij})>0$ for some $\tau_{ij}>0$. After time $\tau_{ij}$ variable $q_{ij}$ becomes non-zero and stays that way forever (according to the argument we made earlier in the proof). Then no $q_{ij}$ escapes from the boundary after time $t+\tau^{*}$, where $\tau^{*}=\max_{i,j}\tau_{ij}$. 

Now we prove the second part of the Lemma. Similar to the previous argument, once the $p_{j}(t)>0$ on the invariant set, $p_{j}(t+\tau)>0$ for all $\tau>0$. It remains to show that $p_{j}(t)> 0$ on the invariant set. But, if $p_{j}(t)=0$ were true, then $\dot{q}_{ij}>0$ would imply $\lim_{t \rightarrow \infty}q_{ij}(t)=\infty $, which would violate (\ref{eqn:sale3}).

\subsection{Proof of Lemma \ref{lem:BpequalsE} --- Proof that $\mathcal{B}\boldsymbol{p}=E$}
\label{app:stackMatricesPrice}
Since $C\boldsymbol{q}(t)=C\boldsymbol{q}^{*}$ is a constant on $V_{L}$, then taking the time derivative of both sides gives $C\dot{\boldsymbol{q}}(t)=0$. Then:
\begin{align*}
0= & \:C\dot{\boldsymbol{q}}(t)=C\hat{K}^{q} \left(\boldsymbol{f}^{*}-A\boldsymbol{p}(t)\right) \\
\Leftrightarrow\: &C\hat{K}^{q} A\boldsymbol{p}(t)=C\hat{K}^{q}  \boldsymbol{f}^{*}, 
\end{align*}
which means that $C\hat{K}^{q} A\boldsymbol{p}(t) $ equals a constant on the invariant set. Taking the time derivative one more time yields
\begin{align*}
0=C\hat{K}^{q} A\dot{\boldsymbol{p}(t)}=&C\hat{K}^{q} A K^{p}\left(A^{T}\boldsymbol{q}(t)-Q\right)\\
\Leftrightarrow\: C\hat{K}^{q} A K^{p}A^{T}\boldsymbol{q}(t)=&C\hat{K}^{q} A K^{p}Q=\text{constant}
\end{align*}
Repeating the derivative operation $2n+1 $ times will yield 
\begin{align*}
C\hat{K}^{q} A (K^{p}A^{T}\hat{K}^{q} A)^{n}\boldsymbol{p}=&C\hat{K}^{q} A (K^{p}A^{T}\hat{K}^{q} A)^{n-1}K^{p}A^{T}\hat{K}^{q}\boldsymbol{f}^{*}, \\
=C\hat{K}^{q} A D^{n}\boldsymbol{p}=&C\hat{K}^{q} A D^{n-1}K^{p}A^{T}\hat{K}^{q}\boldsymbol{f}^{*}=\text{const}
\end{align*}
where we defined $D=K^{p}A^{T}\hat{K}^{q} A$. Note that $D^{n}$ here stands for ``$D$ to the $n^{th}$ power''. 
Let $B=C\hat{K}^{q} A$, then we can write:
\begin{align*}
\left[\begin{array}{cccc}B &BD &\cdots &BD^{J-1}\end{array}\right]^{T}\boldsymbol{p}(t)=\mathcal{B}\boldsymbol{p}(t)=\text{constant}, 
\end{align*}
which completes the proof.

\subsection{Proof of Lemma~\ref{lem:rankJ} }
\label{app:proofMatrixMinusEVfullRank}
Consider matrix $D=K^{p}A^{T}\hat{K}^{q} A$. It can be verified that
\begin{align*}
D=\left[\begin{array}{ccc}k^{p}_{1}(\sum_{i \in\mathcal{I}_{1}}^{¥}k_{i1})  & \cdots & 0 \\\vdots & \ddots & \vdots \\0  & \cdots & k^{p}_{J}(\sum_{i \in \mathcal{I}_{J}}^{¥}k_{iJ})\end{array}\right] \\ \stackrel{(a)} {=}
k\left[\begin{array}{cccc}k^{p}_{1}|\mathcal{I}_{1}| & 0 & \cdots & 0 \\0 & k^{p}_{2}|\mathcal{I}_{2}| & \cdots & 0 \\\vdots & \vdots & \ddots & \vdots \\0 & 0 & \cdots & k^{p}_{J}|\mathcal{I}_{J}|\end{array}\right],
\end{align*}
where $\mathcal{I}_{j}=\{i \in \mathcal{I}: q_{ij}(t)>0\}$ is the set of users of provider $j$, where (a) follows from $k^{q}_{ij}=k$. Without loss of generality, take $k=1$. Now, assuming that $k^{p}_{j}$'s are not integer multiples of each other, we can see that at most one row of $D-\lambda \mathbb{I}$ can be all-zero vector for any eigenvalue $\lambda$ of $D$ (indeed, exactly one since $D$ is a diagonal matrix). So, we can always choose $J-1$ rows of $D-\lambda \mathbb{I}$ that are non-zero and also linearly independent. An example of $G'$ is:

\begin{align*}
\left[\begin{array}{cccc}
k^{p}_{1}|\mathcal{I}_{1}|-\lambda_{J} &  \cdots & 0& 0 \\\vdots &  \ddots & \vdots& \vdots \\0  & \cdots & k^{p}_{J-1}|\mathcal{I}_{J-1}|-\lambda_{J}& 0 \\ k^{q}_{i1}c_{i1}\mathbbm{1}_{\mathbb{R}^{+}}(q_{i1})  & \cdots & k^{q}_{iJ-1}c_{iJ-1}\mathbbm{1}_{\mathbb{R}^{+}}(q_{iJ-1})& k^{q}_{iJ}c_{iJ} \end{array} \right]
\end{align*}

{Without loss of generality, assume that we choose rows $1$ through $J-1$ (i.e., we consider the $J^{th}$ eigenvalue $\lambda_{J}=k^{p}_{J}|\mathcal{I}_{J}|$). Now, consider matrix $B'=C{K}^{q} A$. It can be verified that $B'_{i,j}=k^{q}_{ij}c_{i,j}$. The entries of matrix $B=C\hat{K}^{q} A$ are then $k^{q}_{ij}c_{ij}\mathbbm{1}_{(0,\infty)}(q_{ij}(t))$. Since $q_{iJ}(t)>0$ for at least one $i=\hat{\imath}$,  $k^{q}_{\hat{\imath}J}c_{\hat{\imath}J}\mathbbm{1}_{(0,\infty)}(q_{\hat{\imath}J}(t))=k^{q}_{\hat{\imath}J}c_{\hat{\imath}J}$, so the $\hat{\imath}$ row has an entry at the $J^{th}$ column. We then append this row to the $J-1$ rows we took from $D-\lambda_{J} \mathbbm{I}$, forming a submatrix of $G$, denoted $G'$. The submatrix $G'$ is then a lower triangular matrix.}
Lower triangular matrices have full rank, in this case $J$. Since the submatrix of $G$ has the full column rank, then $G$ also has the full column rank $J$. The procedure works in the same way for any eigenvalue $\lambda_{j}$ of $D$ we may choose, since we can always relabel the providers. This completes the proof.

\bibliography{IEEEabrv,biblio}

\begin{thebibliography}{10}
\providecommand{\url}[1]{#1}
\csname url@samestyle\endcsname
\providecommand{\newblock}{\relax}
\providecommand{\bibinfo}[2]{#2}
\providecommand{\BIBentrySTDinterwordspacing}{\spaceskip=0pt\relax}
\providecommand{\BIBentryALTinterwordstretchfactor}{4}
\providecommand{\BIBentryALTinterwordspacing}{\spaceskip=\fontdimen2\font plus
\BIBentryALTinterwordstretchfactor\fontdimen3\font minus
  \fontdimen4\font\relax}
\providecommand{\BIBforeignlanguage}[2]{{%
\expandafter\ifx\csname l@#1\endcsname\relax
\typeout{** WARNING: IEEEtran.bst: No hyphenation pattern has been}%
\typeout{** loaded for the language `#1'. Using the pattern for}%
\typeout{** the default language instead.}%
\else
\language=\csname l@#1\endcsname
\fi
#2}}
\providecommand{\BIBdecl}{\relax}
\BIBdecl

\bibitem{FCC:2008uq}
``{FCC} press release: {FCC} adopts rules for unlicenced use of television
  white spaces,'' November 2008.

\bibitem{Leyffer:2005pi}
S.~Leyffer and T.~Munson, ``Solving multi-leader-follower games,''
  \emph{Preprint ANL/MCS-P1243-0405}, 04 2005/04/22.

\bibitem{Wang:2008jt}
J.~Wang, D.~Chiu, and J.~Lui, ``A game--theoretic analysis of the implications
  of overlay network traffic on {ISP} peering,'' \emph{Computer Networks},
  vol.~52, no.~15, pp. 2961--2974, October 2008.

\bibitem{Kelly:1998}
F.~Kelly, A.~Maulloo, and D.~Tan, ``Rate control for communication networks:
  Shadow prices, proportional fairness and stability,'' \emph{Journal of the
  Operational Research Society}, vol.~49, pp. 237--252, 1998.

\bibitem{Lin:2006bc}
X.~Lin and N.~B. Shroff, ``Utility maximization for communication networks with
  multipath routing,'' \emph{IEEE Transactions on Automatic Control}, vol.~51,
  no.~5, pp. 766--781, 2006.

\bibitem{Voice:2006it}
T.~Voice, ``Stability of congestion control algorithms with multi-path routing
  and linear stochastic modelling of congestion control,'' Ph.D. dissertation,
  University of Cambridge, Cambridge, UK, May 2006.

\bibitem{Bertsekas:1999yq}
D.~P. Bertsekas, \emph{Nonlinear Programming}.\hskip 1em plus 0.5em minus
  0.4em\relax {Athena Scientific}, Sept. 1999.

\bibitem{Fudenberg:1991}
D.~Fudenberg and J.~Tirole, \emph{Game Theory}.\hskip 1em plus 0.5em minus
  0.4em\relax MIT Press, 1991.

\bibitem{Mas-Colell:1995}
A.~Mas-Colell, M.~D. Whinston, and J.~R. Green, \emph{Microconomic
  Theory}.\hskip 1em plus 0.5em minus 0.4em\relax Oxford University Press,
  1995.

\bibitem{Rad:2009sf}
A.~Rad, J.~Huang, M.~Chiang, and V.~Wong, ``Utility-optimal random access
  without message passing,'' \emph{IEEE Transactions on Wireless
  Communications}, vol.~8, no.~3, pp. 1073--1079, 2009.

\bibitem{Chen:2008la}
M.~Chen and J.~Huang, ``Optimal resource allocation for {OFDM} uplink
  communication: A primal-dual approach,'' in \emph{CISS}, 2008, pp. 926--931.

\bibitem{Khalil:2002uq}
H.~Khalil, \emph{Nonlinear Systems}.\hskip 1em plus 0.5em minus 0.4em\relax
  Prentice Hall, 2002.

\bibitem{Chen:1998uq}
C.-T. Chen, \emph{Linear System Theory and Design}, 3rd~ed.\hskip 1em plus
  0.5em minus 0.4em\relax Oxford University Press, 1998.

\bibitem{Wong:1997rw}
D.~Wong and T.~J. Lim;, ``Soft handoffs in {CDMA} mobile systems,'' \emph{IEEE
  Personal Communications}, vol.~4, no.~6, pp. 6--17, December 1997.

\bibitem{Maille:2008nh}
P.~Maille and B.~Tuffin, ``Analysis of price competition in a slotted resource
  allocation game,'' in \emph{INFOCOM}, 2008, pp. 888--896.

\bibitem{Shen:2004qa}
H.~Shen and T.~Basar, ``Differentiated internet pricing using a hierarchical
  network game model,'' \emph{ACC}, vol.~3, pp. 2322--2327, January 2004.

\bibitem{fierce2010:fk}
\BIBentryALTinterwordspacing
{Fierce Wireless}, ``Is usage based pricing inevitable?'' [Online]. Available:
  \url{http://www.fiercewireless.com/story/usage-based-pricing-inevitable/2010%
-02-03}
\BIBentrySTDinterwordspacing

\bibitem{engadget2010:fk}
\BIBentryALTinterwordspacing
Engadget, ``{AT\&T} makes sweeping changes to data plans, i{P}hone tethering
  coming at {OS} 4 launch.'' [Online]. Available:
  \url{http://www.engadget.com/2010/06/02/atandt-makes-sweeping-changes-to-dat%
a-plans-iphone-tethering-comi/}
\BIBentrySTDinterwordspacing

\bibitem{Saraydar:2002}
C.~U. Saraydar, N.~B. Mandayam, and D.~J. Goodman, ``Efficient power control
  via princing in wireless data networks,'' \emph{IEEE Transactions on
  Communications}, vol.~50, no.~2, pp. 291--303, February 2002.

\bibitem{Marbach:2002fk}
P.~Marbach and R.~Berry, ``Downlink resource allocation and pricing for
  wireless networks,'' in \emph{INFOCOM}, vol. vol. 3, 2002, pp. 1470--1479.

\bibitem{Chiang:2004}
M.~Chiang and J.~Bell, ``Balancing supply and demand of bandwidth in wireless
  cellular networks: Utility maximization over powers and rates,'' in
  \emph{INFOCOM}, vol.~4, 2004, pp. 2800--2811.

\bibitem{Acemoglu:2004bv}
D.~Acemoglu, A.~Ozdaglar, and R.~Srikant, ``The marginal user principle for
  resource allocation in wireless networks,'' in \emph{CDC}, vol.~2, 2004, pp.
  1544--1549.

\bibitem{Musacchio:2006tg}
J.~Musacchio and J.~Walrand, ``{WiFi} access point pricing as a dynamic game,''
  \emph{IEEE/ACM Transactions on Networking}, vol.~14, no.~2, 04 2006/04/01.

\bibitem{Jiang:2008rm}
L.~Jiang, S.~Parekh, and J.~Walrand, ``Base station association game in
  multi-cell wireless networks,'' \emph{WCNC}, pp. 1616--1621, 2008.

\bibitem{Adlakha:2007rw}
S.~Adlakha, R.~Johari, and A.~J. Goldsmith, ``Competition in wireless systems
  via {Bayesian} interference games,'' \emph{CoRR}, 2007, informal publication.

\bibitem{Etkin:2005uq}
R.~Etkin, A.~Parekh, and D.~Tse, ``Spectrum sharing for unlicensed bands,'' in
  \emph{DySPAN}, 2005, pp. 251--258.

\bibitem{Huang:2006}
J.~Huang, R.~A. Berry, and M.~L. Honig, ``Distributed interference compensation
  for wireless networks,'' \emph{IEEE Journal on Selected Areas in
  Communications}, vol.~24, no.~5, pp. 1074--1084, 2006.

\bibitem{Kalathil:2010vn}
D.~Kalathil and R.~Jain, ``Spectrum sharing through contracts,'' in
  \emph{DySPAN}, 2010.

\bibitem{Zhou:2005ys}
C.~Zhou, M.~L. Honig, and S.~Jordan, ``Utility-based power control for a
  two-cell {CDMA} data network,'' \emph{IEEE Transactions on Wireless
  Communications}, vol.~4, no.~6, pp. 2764--2776, 2005.

\bibitem{Grokop:2008}
L.~Grokop and D.~N. Tse, ``Spectrum sharing between wireless networks,'' in
  \emph{INFOCOM}, April 2008, pp. 201 -- 205.

\bibitem{Zemlianov:2005}
A.~Zemlianov and G.~de~Veciana, ``Cooperation and decision-making in a wireless
  multi-provider setting,'' in \emph{INFOCOM}, vol.~1, 2005, pp. 386--397.

\bibitem{Sengupta:2007}
S.~Sengupta, M.~Chatterjee, and S.~Ganguly, ``An economic framework for
  spectrum allocation and service pricing with competitive wireless service
  providers,'' in \emph{DySPAN}, 2007, pp. 89--98.

\bibitem{Jia:2008}
J.~Jia and Q.~Zhang, ``Competitions and dynamics of duopoly wireless service
  providers in dynamic spectrum market,'' in \emph{MobiHoc}, New York, NY, USA,
  2008, pp. 313--322.

\bibitem{Ileri:2005a}
O.~Ileri, D.~Samardzija, T.~Sizer, and N.~B. Mandayam, ``Demand responsive
  pricing and competitive spectrum allocation via a spectrum server,'' in
  \emph{DySPAN}, 2005, pp. 194--202.

\bibitem{Acharya:2009qd}
J.~Acharya and R.~D. Yates, ``Service provider competition and pricing for
  dynamic spectrum allocation,'' in \emph{GameNets}, 2009, pp. 190--198.

\bibitem{Inaltekin:2007}
H.~Inaltekin, T.~Wexler, and S.~B. Wicker, ``A duopoly pricing game for
  wireless {IP} services,'' in \emph{SECON}, 2007, pp. 600--609.

\bibitem{Niyato:2009tg}
D.~Niyato, E.~Hossain, and Z.~Han, ``Dynamics of multiple-seller and
  multiple-buyer spectrum trading in cognitive radio networks: A game-theoretic
  modeling approach,'' \emph{IEEE Transactions on Mobile Computing}, vol.~8,
  no.~8, pp. 1009--1022, 2009.

\bibitem{Xi:2008nx}
Y.~Xi and E.~M. Yeh, ``Pricing, competition, and routing for selfish and
  strategic nodes in multi-hop relay networks,'' \emph{INFOCOM}, pp.
  1463--1471, 13-18 April 2008.

\bibitem{Gao:2008kx}
L.~Gao and X.~Wang, ``A game approach for multi-channel allocation in multi-hop
  wireless networks,'' in \emph{MobiHoc '08}.\hskip 1em plus 0.5em minus
  0.4em\relax New York, NY, USA: ACM, 2008, pp. 303--312.

\bibitem{Boyd:2004}
S.~Boyd and L.~Vandenberghe, \emph{Convex Optimization}.\hskip 1em plus 0.5em
  minus 0.4em\relax Cambridge University Press, 2004.

\end{thebibliography}

\bibliographystyle{IEEEtran}

\end{document}